\documentclass[reqno,12pt,epsfig]{amsart} 

\date{\today}

\usepackage{amsmath}
\usepackage{amsfonts}
\usepackage{amssymb}
\usepackage{amsthm}
\usepackage{amstext}
\usepackage{enumerate}
\usepackage{graphicx}
\usepackage{graphics}
\usepackage[usenames]{color}
\usepackage{appendix}

\newcommand{\1}{{\rm 1\hspace*{-0.4ex}%
\rule{0.1ex}{1.52ex}\hspace*{0.2ex}}}

\usepackage{epsfig}
\usepackage{pst-tree}
\usepackage{ifthen}
\usepackage{psfrag}
\newcommand\ersetze[3][1]{\psfrag{#2}[Bl][Bl][#1][0]{#3}}

\newcommand{\bx}{\mathbf{x}} 
\newcommand{\by}{\mathbf{y}}
\newcommand{\bz}{\mathbf{z}}
\newcommand{\bp}{\mathbf{p}}
\newcommand{\ba}{\mathbf{a}}
\newcommand{\bb}{\mathbf{b}}

\newcommand{\rz}{\mathbb{R}}
\newcommand{\N}{\mathbb{N}}
\newcommand{\C}{\mathbb{C}}
\newcommand{\R}{\mathbb{R}}
\newcommand{\p}{\partial}

\newcommand{\supp}{{\operatorname{supp}}}

\newcommand{\dist}{{\operatorname{dist}}}
\newcommand{\Arg}{{\operatorname{Arg}}}
\newcommand{\Log}{{\operatorname{Log}}}

\theoremstyle{plain}
\newtheorem{theorem}{Theorem}[section]{\bf}{\it}
\newtheorem{prop}[theorem]{Proposition}{\bf}{\it}
\newtheorem{lemma}[theorem]{Lemma}{\bf}{\it}
{\bf}{\it}
\newtheorem{corollary}[theorem]{Corollary}{\bf}{\it}
\theoremstyle{definition}
\newtheorem*{acknowledgement}{Acknowledgement} 
\newtheorem{remark}[theorem]{Remark}{\it}{\rm}
\newtheorem{defn}[theorem]{Definition}{\bf}{\rm}
{\rm}{\rm}

\newenvironment{pf}{\par\medskip\noindent\textit{Proof}:\,}{\hspace*{\fill}\qed\medskip\par\noindent} 
\newenvironment{pf*}[1]{\par\medskip\noindent\textit{#1}\,:}{\hspace*{\fill}\qed\medskip\par\noindent}   

\title[Analyticity of pseudorelativistic Hartree--Fock orbitals]{Real
  analyticity away 
  from the nucleus of pseudorelativistic Hartree--Fock orbitals} 
\author[A. Dall'Acqua, S. Fournais, T. {\O}. S{\o}rensen, and
E. Stockmeyer]{Anna Dall'Acqua, S{\o}ren Fournais, Thomas \O stergaard
  S\o rensen, and Edgardo Stockmeyer} 

\thanks{\copyright\ 2011 by the
       authors. This article may be reproduced, in its entirety, for
       non-commercial purposes.}

\address[Anna Dall'Acqua]
{Institut f\"ur Analysis und Numerik,
Fakult\"at f\"ur Mathematik,
Otto-von-Guericke Universit\"at,
Postfach 4120,
D-39016 Magdeburg, Germany.}
\email{anna.dallacqua@ovgu.de}

\address[S. Fournais]
        {Department of Mathematical Sciences, 
         University of Aarhus, 
         Ny Munkegade 118,
         DK-8000 \AA rhus C, Denmark.} 
\email{fournais@imf.au.dk}           

\address[S. Fournais on leave from]
        {CNRS and Laboratoire de
         Math\'{e}matiques d'Orsay, 
         Univ Paris-Sud, 
         Orsay CEDEX, F-91405, France.} 
\address[Thomas {\O}stergaard S{\o}rensen]
{Department of Mathematics,
Imperial College London,
Huxley Building,
180 Queen's Gate,
London SW7 2AZ, UK.}
\email{t.sorensen@imperial.ac.uk}

\address[Thomas {\O}stergaard S{\o}rensen (present address)]
{
Mathematisches Institut,
Universit\"at M\"unchen,
Theresienstra\ss e 39,
D-80333 Munich, Germany.}
\email{sorensen@mathematik.uni-muenchen.de}

\address[Edgardo Stockmeyer]
{Mathematisches Institut,
Universit\"at M\"unchen,
Theresienstra\ss e 39,
D-80333 Munich, Germany.}
\email{stock@mathematik.uni-muenchen.de}

\begin{document}

\begin{abstract}  
  We prove that the Hartree--Fock orbitals of pseudorelativistic atoms,
  that is, atoms where the kinetic energy of the electrons is given by
  the pseudo\-relativistic operator $\sqrt{{}-\Delta+1}-1$, 
  are real analytic away from the origin. As a
    consequence, the quantum mechanical ground state of such atoms is
    never a Hartree-Fock state.

Our proof is inspired by the classical proof of analyticity by nested
balls of  Morrey and Nirenberg \cite{Morrey-Nirenberg}. 
However, the technique has to be adapted to take 
care of the non-local pseudodifferential operator, the
singularity of the potential at the origin, and the non-linear terms
in the equation. 
\end{abstract}

\maketitle

\section{Introduction and results}\label{sec:intro}

In a recent paper \cite{relHF}, three of the present authors studied the
Hartree--Fock model for pseudorelativistic atoms, and proved the
existence of Hartree--Fock minimizers. Furthermore, they proved that
the corresponding Hartree--Fock orbitals (solutions to the associated
Euler-La\-grange equation) are smooth away from the
nucleus, and that they decay exponentially. 
In this paper we
prove that all of these orbitals are, in fact, real analytic away from the
origin. 
Apart from intrinsic 
mathematical interest, analyticity of
solutions has important consequences. For example, in the
non-relativistic case, the analyticity of the orbitals 
was used in 
\cite{Friesecke, Lewin-thesis} to prove that the quantum mechanical ground state is never
a Hartree--Fock state (or, more generally, is never a finite
  linear combination of Slater determinants). A direct consequence of
  our main regularity result is that this also holds in the
  pseudorelativistic case.   
Our proof also shows that any \(H^{1/2}\)-solution
\(\varphi:\R^3\to\C\) to
the non-linear equation  
\begin{align}\label{eq:non-lin1}
  (\sqrt{{}-\Delta+1})\varphi -
  \frac{Z}{|\cdot|}\varphi\pm\big(|\varphi|^2*|\cdot|^{-1}\big)\varphi
  =\lambda\varphi 
\end{align}
which is smooth away from \(\bx=0\), is in fact real analytic
there. As will be clear from the proof, our method yields the
same result for solutions to equations of the form
\begin{align}\label{eq:non-lin2}
  ({}-\Delta+m)^{s}\varphi + V\varphi+|\varphi|^k\varphi=\lambda\varphi\,,
\end{align}
where \(V\) has a finite number of point singularities (but is analytic
elsewhere), under certain conditions on \(m, s, V\), and \(k\) (see
Remark~\ref{rem:main-thm} below).
We believe this result is of independent interest, but stick
concretely to the case of pseudorelativistic Hartree--Fock orbitals,
since this was the original motivation for the present work.

We consider a model for an atom with $N$ electrons and nuclear charge
$Z$ (fixed at the origin), where the kinetic energy of the electrons
is described by 
the expression $\sqrt{(|\bp|c)^2+(mc^2)^2}-mc^2$. This model takes into
account some (kinematic) relativistic effects; in units where
\(\hbar=e=m=1\), the Hamiltonian becomes
\begin{align}\label{Hamiltonian}
  H=\sum_{j=1}^{N}\alpha ^{-1}\Big\{T(-{\rm i}\nabla_{j})-V({\bx}_{j})
  \Big\}
  +\sum_{1\leq i<j\leq N}\frac{1}{|\bx_{i}-\bx_{j}|}\,,  
\end{align}
with
\(T({\bp})=E({\bp})-\alpha^{-1}=\sqrt{|{\bp}|^2+\alpha^{-2}}-\alpha^{-1}\)  
and \(V({\bx})=Z\alpha/|{\bx}|\). Here, $\alpha $ is Sommerfeld's fine
structure constant; physically, \(\alpha\simeq1/137\).

The operator $H$ acts on a dense subspace of the
$N$-particle Hilbert space $\mathcal{H}_{F}=\wedge 
_{i=1}^{N}L^{2}(\mathbb{R}^{3})$ of antisymmetric
functions.
(We will not consider spin since it is irrelevant for our discussion.)
It is bounded from below on this
subspace if and only if 
\(Z\alpha\le 2/\pi\) (see \cite{Lieb-Yau};
for a number of other works on this operator, see
\cite{CarmonaEtAl,
  Daubechies-Lieb1, FeffermanEtAl, Herbst, LewisEtAl,
  Nardini1,Weder, Zhislin1}).

The \textit{(quantum) ground state energy} 
is the infimum of 
the quadratic form \(\mathfrak{q}\) defined by \(H\),
over the subset of elements of norm \(1\) of the corresponding form
domain.
Hence, it 
coincides with the infimum of the
spectrum of $H$ considered as an operator acting in
$\mathcal{H}_{F}$. A corresponding minimizer is called a 
  \textit{(quantum) ground state} of \(H\).

In the Hartree--Fock approximation, instead of minimizing the
quadratic form \(\mathfrak{q}\) in the entire $N$-particle space
\(\mathcal{H}_{F}\), one 
restricts to wavefunctions \(\Psi\) which are pure wedge products,
also called Slater determinants: 
\begin{align}\label{slater}
  \Psi (\bx_{1},\dots,\bx_{N})
  =\frac{1}{\sqrt{N!}}\,\det(u_{i}(\bx_{j}))_{i,j=1}^{N}
  \,, 
\end{align}
with $\{u_{i}\}_{i=1}^N$ orthonormal in
$L^{2}(\mathbb{R}^{3})$ (called {\it orbitals}). Notice 
that this way, \(\Psi\in\mathcal{H}_{F}\) and
$\|\Psi\|_{L^{2}(\mathbb{R}^{3N})}=1$.   

The {\it Hartree--Fock ground state energy} is the infimum of the
quadratic form \(\mathfrak{q}\) defined by \(H\) over such Slater
determinants: 
\begin{align}\label{eq:HF-energy}
   E^{{\rm HF}}(N,Z,\alpha):=
   \inf \{\,\mathfrak{q}(\Psi,\Psi) \,|\, \Psi \ \text{Slater
     determinant}\,\}\,. 
\end{align}
 Inserting \(\Psi\) of the
form in \eqref{slater} into \(\mathfrak{q}\)
formally yields
\begin{align}
   \label{eq:HF-functional}\nonumber
   \mathcal{E}^{\rm HF}(u_1,\ldots,&u_N):=\mathfrak{q}(\Psi,\Psi) 
   \\\nonumber&=\alpha^{-1}\sum_{j=1}^{N}\int_{\R^3}\big\{\,
   \overline{u_{j}(\bx)}\,[T(-{\rm i}\nabla)u_j](\bx)
   -V(\bx)
   |u_j(\bx)|^2\big\}\,d\bx
   \\\nonumber&\   
   +\frac{1}{2}\sum_{1\le i,j\le N}\int_{\R^3}\int_{\R^3}
   \frac{|u_{i}(\bx)|^2|u_{j}(\by)|^2}{|\bx-\by|}
   \,d\bx d\by
   \\&\ -\frac{1}{2}\sum_{1\le i,j\le N}\int_{\R^3}\int_{\R^3}
   \frac{\overline{u_{j}(\bx)}u_{i}(\bx)\overline{u_{i}(\by)}u_{j}(\by)}{|\bx-\by|} 
   \,d\bx d\by\,.
 \end{align}
In fact, \(u_{i}\in H^{1/2}(\R^{3})\), \(1\le i\le N\), is needed for
this to be well-defined (see Section~\ref{sec:main estimate} for a
detailed discussion), and so
\eqref{eq:HF-energy}--\eqref{eq:HF-functional} can be written
\begin{align}\label{eq:HF-min}
   E^{{\rm HF}}(N,Z,\alpha)
   &=\inf\{\,\mathcal{E}^{\rm HF}(u_1,\ldots,u_N)\,|\, 
  (u_1,\ldots,u_N)\in \mathcal{M}_N\}\,,
  \\
  \mathcal{M}_N&=\big\{\, 
  (u_1,\ldots,u_N)\in [H^{1/2}(\R^3)]^N
  \,\big|\
  (u_i,u_j)=\delta_{ij}\,\big\}\,. \label{eq:HF-constraints}
\end{align}
 Here,  $(\ ,\ )$ denotes the scalar product in
 $L^2(\mathbb{R}^3)$.
The existence of minimizers for the problem
\eqref{eq:HF-min}--\eqref{eq:HF-constraints} was proved in
\cite{relHF} when $Z>N-1$ and $Z\alpha <2/\pi$. (Note that such
minimizers are generally not unique since \(\mathcal{E}^{\rm HF}\) is
not convex; see \cite{KS-HF}). The existence of
infinitely many distinct critical points 
of the functional  \(\mathcal{E}^{\rm
  HF}\) on  \(\mathcal{M}_N\) was proved recently (under the same
conditions) in \cite{EnstedtMelgaard}.  

The Euler--Lagrange equations of
the problem \eqref{eq:HF-min}--\eqref{eq:HF-constraints} are the {\it
  Har\-tree--Fock equations},
\begin{align}\label{eq:HF}
  \big[\big(T(-{\rm i}&\nabla)-V\big)
   \varphi_{i}\big](\bx)
  +\alpha\Big(\sum_{j=1}^{N}\int_{\R^3}\frac{|\varphi_j(\by)|^2}{|\bx-\by|}
  \,d\by\Big)\varphi_i(\bx) 
  \\&-
  \alpha\sum_{j=1}^{N}
  \Big(\int_{\R^3}
  \frac{\overline{\varphi_j(\by)}\varphi_i(\by)}{|\bx-\by|}\,d\by\Big)\varphi_j(\bx) 
  =\varepsilon_i\varphi_i(\bx)\ , \quad 1\le i\le N\,.
  \nonumber
\end{align}
Here, the \(\varepsilon_i\)'s are the Lagrange multipliers of the
orthonormality constraints in \eqref{eq:HF-constraints}. (Note that the
naive Euler--Lagrange equations are more complicated than 
\eqref{eq:HF}, but can be transformed to \eqref{eq:HF}; see \cite{KS-HF}.) 
Note that \eqref{eq:HF} can be re-formulated as 
\begin{align}
  \label{eq:HF-operator eigen}
  h_{{\boldsymbol{\varphi}}}\varphi_i  = \varepsilon_i\varphi_i\ , \quad
  1\le i\le N\,,
\end{align}
with \(h_{{\boldsymbol{\varphi}}}\) the {\it Hartree--Fock operator associated
  to \({\boldsymbol{\varphi}}=\{\varphi_1,\ldots,\varphi_N\}\)}, formally
given by
\begin{align}
  \label{eq:HF-operator}
  h_{{\boldsymbol{\varphi}}}u=
  [T(-{\rm i}\nabla)-V]u+
  \alpha R_{{\boldsymbol{\varphi}}}u-\alpha K_{{\boldsymbol{\varphi}}}u\,,
\end{align}
where \(R_{{\boldsymbol{\varphi}}}u\) is the
{\it direct interaction}, given by the multiplication operator defined
by   
\begin{equation}\label{Rgamma}
  R_{{\boldsymbol{\varphi}}}(\bx):=\sum_{j=1}^{N}\int_{\mathbb{R}^3}
  \frac{|\varphi_j(\by)|^2}{|\bx-\by|}\,d
  \by\,
\end{equation}
and \(K_{{\boldsymbol{\varphi}}}u\) is the {\it exchange term}, given by
the integral operator
\begin{equation}\label{Kgamma}
  (K_{{\boldsymbol{\varphi}}}u)(\bx)=\sum_{j=1}^{N}\Big(\int_{\mathbb{R}^3}
  \frac{\overline{\varphi_j(\by)}u(\by)}{|\bx-\by|}\,d
  \by\Big)\varphi_{j}(\bx)\,.
\end{equation}
The equations \eqref{eq:HF}
(or equivalently \eqref{eq:HF-operator eigen}) are called the {\it
  self-consistent Hartree--Fock equations}. 
One has that \(\sigma_{\rm
  ess}(h_{{\boldsymbol{\varphi}}})=[0,\infty)\) and that, when in addition \(N<Z\), 
the operator
\(h_{{\boldsymbol{\varphi}}}\) has infinitely many eigenvalues in
\([-\alpha^{-1},0)\)
(see
\cite[Lemma~2]{relHF}; the argument given there holds for any
 \(\boldsymbol{\varphi}=\{\varphi_{1},\ldots,\varphi_{N}\}\),
   \(\varphi_ {i}\in H^{1/2}(\R^3)\), as long as \(Z\alpha<2/\pi\)).
If 
\((\varphi_1,\ldots,\varphi_N)\in\mathcal{M}_N\) 
is a minimizer for
the problem \eqref{eq:HF-min}--\eqref{eq:HF-constraints}, 
then the \(\varphi_i\)'s solve \eqref{eq:HF-operator eigen} with
\(\varepsilon_1\le\varepsilon_2\le \cdots\le\varepsilon_N<0\) the
\(N\) {\it lowest} eigenvalues of the ope\-rator \(h_{{\boldsymbol{\varphi}}}\)
\cite{relHF}.

In \cite{relHF} it was proved that solutions
\(\{\varphi_{1},\ldots,\varphi_{N}\}\) to \eqref{eq:HF}---and,
more generally, all eigenfunctions of the corresponding Hartree--Fock
operator \(h_{{\boldsymbol{\varphi}}}\)---are smooth away from \(\bx=0\) (the
singularity of \(V\)), and that 
(for the \(\varphi_{i}\)'s for which \(\varepsilon_{i}<0\))
they decay exponentially. 
(The solutions studied in \cite{relHF} came from a minimizer of
\(\mathcal{E}^{\rm HF}\), but the proof trivially 
extends to the solutions
\(\{{\boldsymbol{\varphi}}_{n}\}_{n\in\N}
=\big\{\{\varphi_1^{n},\ldots,\varphi_N^{n}\}\big\}_{n\in\N}\) 
to  \eqref{eq:HF} 
found in \cite{EnstedtMelgaard}, and to all the 
eigenfunctions of
the corresponding Hartree--Fock operators mentioned above). The main
theorem of this paper 
is the following, which completely settles the question of regularity
away from the origin of solutions to the equations \eqref{eq:HF}. 
\begin{theorem}\label{HF2}
Let $Z\alpha <2/\pi $, and let $N\ge2$ be a positive integer such
that $N<Z+1$. 
Let \({\boldsymbol{\varphi}}=\{\varphi_1,\ldots,\varphi_N\}\),
\(\varphi_{i}\in H^{1/2}(\R^{3})\), \(i=1,\ldots,N\), 
be solutions to the pseudorelativistic Hartree--Fock equations in
\eqref{eq:HF}.
 
Then, for \(i=1,\ldots,N\),
   \begin{align}
     \label{eq:regularityBIS}
      \varphi_i\in C^{\omega}(\rz^3\setminus\{0\})\,,
   \end{align}
that is, the Hartree--Fock orbitals are real analytic away from the
origin in \(\rz^3\).
\end{theorem}
\begin{remark}\label{rem:main-thm}
(i) The restrictions \(Z\alpha<2/\pi\), \(N<Z+1\), and \(N\ge2\) are only made to
ensure {\it existence} of \(H^{1/2}\)-solutions to \eqref{eq:HF}. In
fact, our proof proves analyticity away from \(\bx=0\) for
\(H^{1/2}\)-solutions to \eqref{eq:HF} for {\it any} \(Z\alpha\). For
the case \(N=1\), \eqref{eq:HF} reduces to 
\((T-V)\varphi=\varepsilon\varphi\) and our result also holds for
\(H^{1/2}\)-solutions to this equation (see also (iv) and (v) below about
more general \(V\) for which the result also holds for the linear equation). 
More interestingly,
the result also holds for 
\(H^{1/2}\)-solutions to \eqref{eq:non-lin1} (which, strictly
speaking, cannot be obtained from \eqref{eq:HF} by any choice of \(N\)).

(ii) The statement also holds for any eigenfunction of the
associated Hartree--Fock operator given by \eqref{eq:HF-operator}.

(iii) It is obvious from the proof that the theorem holds true if we
include spin. 

(iv)
  As will also be clear from the proof, the
  statement of Theorem~\ref{HF2} (appropriately modified) also holds
  for molecules. More explicitely, for a molecule with \(K\) nuclei of
  charges \(Z_1,\ldots,Z_K\), fixed at \(R_1,\ldots,R_K\in\rz^3\),
  replace \(V\) in \eqref{eq:HF} by 
   \( \sum_{k=1}^{K}V_k\)
  with \(V_k(\bx)=Z_k\alpha/|\bx-R_k|, Z_k\alpha<2/\pi\). Then, for
  \(N<1+\sum_{k=1}^KZ_k\), 
Hartree--Fock
  minimizers exist (see \cite[Remark 1 (viii)]{relHF}), and the corresponding
  Hartree--Fock orbitals  
 are real analytic away from the positions of the nuclei, i.e., belong
 to \(C^{\omega}(\R^3\setminus\{R_1,\ldots,R_K\})\). 

(v) Another approximation to the full quantum mechanical
  problem is the {\it multiconfiguration self-consistent field method}
  (MC-SCF). Here one minimizes the quadratic form  \(\mathfrak{q}\)
  defined by the operator \(H\) given in \eqref{Hamiltonian} (or, more
  generally, with \(V\) from (iv)) over the set of {\it finite} sums
  of Slater determinants instead of only on single Slater
  determinants as in Hartree--Fock theory. 
If minimizers exist they satisfy what is called
the {\it multiconfiguration 
  equations} (MC equations). For more details, see \cite{KS-HF,
  Friesecke, Lewin}. As will be clear from the proof, the statement of
Theorem~\ref{HF2} also holds for 
solutions to these equations.

(vi) In fact, for \(V\) we only need the analyticity of \(V\) away
from finitely many points in \(\R^{3}\), and certain 
integrability properties of 
\(V\varphi_{i}\) in the vicinity of each of these points, and at
infinity; for more details, see Remark~\ref{rem:exp-decay}.

(vii) As will be clear from the proof, the 
  statement of Theorem~\ref{HF2} also holds for other non-linearities
  than the Hartree-Fock term in \eqref{eq:HF}, namely
  \(|\varphi|^{k}\varphi\)
  as in \eqref{eq:non-lin2}
  (for \(k\) even; for \(k\) odd, one needs to take
  \(\varphi^{k+1}\)). The \(L^{p}\)-space in which one needs to study
  the problem (see Proposition~\ref{lemmaestimate} and the description
  of the proof below for details) needs to be chosen depending on
  \(k\) in this case (the larger the \(k\), the larger the \(p\)).

(viii) Also, as will be clear from the proof, the result holds if
\(T(-{\rm i}\nabla)=|\nabla|\) (i.e., \(T(\bp)=|\bp|\)) in
\eqref{eq:HF}. In \eqref{eq:use-equation} below, \(E(\bp)^{-1}\) should
then be replaced by \((|\bp|+1)^{-1}\) (and \ `\(1\)` added to
`\(\alpha^{-1}+\varepsilon_{i}\)'). The only properties of
\(E(\bp)^{-1}\) used are in Lemmas~\ref{bounded-mult-op}
and \ref{normsmooth-Lp}, which follow also for
\((|\bp|+1)^{-1}\) from the same methods 
with minor modifications. Similarly, one can replace
\(T(\bp)\) with
\((-\Delta+\alpha^{-2})^{s}\), \(s\in[1/2,1]\).

(ix) The result of Theorem~\ref{HF2} in the non-relativistic case
(\(T(-{\rm i}\nabla)\) replaced by \(-\alpha\Delta\) in
\eqref{Hamiltonian}) was proved in \cite{Friesecke, Lewin-thesis}; see also
the discussion below. In this case, it is furthermore known
\cite{KS-HF} that, for 
\(\bx\in B_r(0)\) for some \(r>0\), 
\(\varphi_{i}(\bx)=\varphi_{i}^{(1)}(\bx)+|\bx|\varphi_{i}^{(2)}(\bx)\)
with \(\varphi_{i}^{(1)},\varphi_{i}^{(2)}\in C^{\omega}(B_{r}(0))\).
\end{remark}

Combining the argument in \cite{Friesecke, Lewin-thesis} with 
the analyticity away
from the position of the nucleus of solutions to the MC equations 
 (see Remark~\ref{rem:main-thm} (v)) 
we readily obtain the following result.
\begin{theorem}\label{thm:QM-HF}
Let \(\Psi\) be a (quantum) ground state of the operator
\(H\) given in \eqref{Hamiltonian}. Then
\(\Psi\) is not a 
\emph{finite} linear combination of 
Slater determinants.
\end{theorem}
\begin{remark}\label{rem:MCSCF}
The same holds with \(V\) as in Remark~\ref{rem:main-thm} (iv). 
\end{remark}

\noindent{\it Description of the proof of Theorem~\ref{HF2}:}
The proof of Theorem~\ref{HF2} is inspired by 
the standard Morrey-Nirenberg \cite{Morrey-Nirenberg} proof of
analyticity of solutions to general 
(linear) elliptic 
partial differential equations with real analytic coefficients 
by `nested balls'. A good presentation of this
technique can be found in 
\cite{Hormander}. (Other proofs using
a complexification of the coordinates also exist and have been applied
to both linear and non-linear equations; see \cite{Morrey-book} and references
therein.)

In \cite{Hormander} one proves $L^2$-bounds on derivatives of 
 order $k$ of the solution in a 
ball \(B_{r}\) (of some radius \(r\)) around a given point. These bounds
should behave suitably in $k$ in order to make the Taylor series of
the solution converge locally, thereby proving analyticity.

The proof of these bounds is inductive. In fact, for some ball
\(B_{R}\) with \(R>r\), one proves the bounds on all balls
\(B_{\rho}\) with \(r\le\rho\le R\), with the appropriate (with
respect to \(k\)) behaviour in \(R-\rho\). 
The induction basis is provided by standard elliptic estimates.
In the induction step, one has to bound $k+1$ derivatives of the solution
in the ball $B_{\rho}$. To do so,
one divides the difference 
$B_{R}\setminus B_{\rho}$ into $k+1$ nested balls using $k+1$
localization functions with 
successively larger supports. Commuting $m$ of the $k$ derivatives (in the
case of an operator of order $m$) with these localization functions
produces (local) differential operators of order $m-1$, with support in a
larger ball. These local commutator terms are controlled by the induction
hypothesis, since they contain one derivative less. 
For the last term---the
term where no commutators occur---one then uses the
equation.

This approach poses new technical difficulties in our case, 
due to the non-locality of the kinetic energy \(T(\bp)
=\sqrt{{}-\Delta+\alpha^{-2}}-\alpha^{-1}\) and the
non-linearity of the terms $R_{{\boldsymbol{\varphi}}}\varphi_{i}$ and 
$K_{{\boldsymbol{\varphi}}}\varphi_{i}$.

The non-locality of the operator
\(\sqrt{{}-\Delta+\alpha^{-2}}\) implies that, as opposed to the case
of a differential operator, the commutator of the kinetic energy with
a localization function is not localized in the support of the
localization function.
That is, when resorting to proving analyticity by differentiating the
equation, the localization argument described above introduces
commutators which 
are (non-local) pseudodifferential operators. Now the induction hypothesis
does not provide control of these terms. 
Furthermore, it is
far from obvious that the singularity of the potential \(V\) outside
\(B_{R}\) does not influence the regularity in
\(B_{R}\) of the solution through these operators (or rather,
through the non-locality of  
\(\sqrt{{}-\Delta+\alpha^{-2}}\)).
Loosely speaking,
the singularity of the nuclear potential `can be felt
everywhere'.
(Note that if we would not have a (singular) potential
$V$ one could proceed as in \cite{Frank-Lenzmann}
and prove global analyticity
by showing exponential decay of the solutions in Fourier space.)  

We overcome this problem by a new
localization argument which enable us to capture in more detail
the action of high order derivatives on nested balls (manifested in
Lemma~\ref{Edgardo} in Appendix~\ref{localization} below).
This, together
with very explicit bounds on the (smoothing) operators  \(\phi
E(\bp)^{-1}D^{\beta}\chi\) for $\chi$ and $\phi$ with disjoint
supports (see Lemma~\ref{normsmooth-Lp}), are the main ingredients in
solving the problem of nonlocality.
The estimates are on 
\(\phi E(\bp)^{-1}D^\beta\chi\) (not \(\phi E(\bp)D^\beta\chi\)),
since we invert \(E(\bp)\) (turning the equation into an integral operator
equation, see \eqref{eq:use-equation}). 
Our method of proof would also
work in the non-relativistic case, since the integral operators
\((-\Delta+1)^{-1}\) and \(E(\bp)^{-1}\) enjoy similar properties.

The second major obstacle is the (morally cubic) non-linearity of the terms
$R_{{\boldsymbol{\varphi}}}\varphi_{i}$ and  
$K_{{\boldsymbol{\varphi}}}\varphi_{i}$. 

To illustrate the problem, we discuss proving analyticity by the
above method (local \(L^2\)-estimates) for solutions \(u\)
to the equation \(\Delta u=u^3\). When differentiating this equation
(and therefore \(u^3\)), the application of Leibniz' rule introduces a
sum of terms. After using H{\"o}lder's inequality on each term (the
product of three factors, each a number of derivatives on \(u\)), one
needs to use a Sobolev inequality to `get back down to \(L^2\)' in
order to use the induction hypothesis. Summing the many terms, the
needed estimate does not come out (in fact, some Gevrey-regularity
would follow, but not analyticity). 

In the quadratic case this  can be done (that 
is, for the equation \(\Delta u=u^2\) this problem does {\it not}
occur), but in the cubic case, one looses too many derivatives.

The second insight of our proof is that this problem of loss of
derivatives may be
overcome 
by characterizing analyticity by growth of
derivatives in some \(L^p\) with \(p>2\). When working in \(L^p\) for
\(p>2\), the loss of derivatives in the Sobolev inequality mentioned
above is less (as seen in Theorem~\ref{adams}). Choosing \(p\)
sufficiently large allows us to prove the needed estimate. The
operator estimates on \(\phi E(\bp)^{-1}D^\beta\chi\) mentioned above
therefore 
have to be  \(L^p\)-estimates. In fact, using \(L^p-L^q\) estimates,
one can also deal with the problem that the singularity of the nuclear
potential \(V\) `can be felt everywhere'.

Note that taking \(p=\infty\) would
avoid using a Sobolev inequality altogether (\(L^\infty\) being an algebra), but
the needed estimates on  \(\phi E(\bp)^{-1}D^\beta\chi\) cannot hold
in this case. 
For local equations an approach to handle the loss of derivatives (due
to Sobolev inequalities) exists. This was carried out in
\cite{Friedman}, where analyticity of solutions to elliptic partial
differential equations with general analytic non-linearities was
proved. Friedman works in spaces of continuous functions. In this
approach, one needs to have a sufficiently high degree of regularity
of the solution beforehand (it is not proved along the way). Also,
since the elliptic regularity in spaces of continuous functions have
an inherent loss of derivative, one needs to work on a sufficiently
small domain in order for the method to work. We prefer to work in
Sobolev spaces since this is the natural setting for our equation and
since the needed estimates on the resolvent are readily obtained in
these spaces. 

For an alternative method of proof (one 
 {\it fixed} localization function, to the power \(k\), and estimating
 in a higher order Sobolev space (instead of in \(L^2\)) which is also
 an algebra), see Kato \cite{Kato-paper} (for the equation \(\Delta
 u=u^2)\) and Hashimoto \cite{Hashimoto} (for general second order non-linear
 analytic PDE's).

Additional technical
difficulties occur due to the fact that the cubic terms, 
\(R_{\boldsymbol{\varphi}}\varphi_{i}\) and
\(K_{\boldsymbol{\varphi}}\varphi_{i}\), 
are actually non-local.

Note that in the
proof that {\it non}-relativistic Hartree-Fock orbitals are analytic
away from the positions of the nuclei (see \cite{Friesecke, Lewin}),
the non-linearities are dealt with by 
cleverly re-writing the Hartree-Fock equations as 
a system. One introduces new functions
\(\phi_{i,j}=[\varphi_i\overline{\varphi_j}]*|\cdot|^{-1}\),
which satisfy \(\,-\Delta
\phi_{i,j}=4\pi\varphi_i\overline{\varphi_j}\). This 
eliminates the terms \(R_{{\boldsymbol{\varphi}}}\varphi_i,
K_{{\boldsymbol{\varphi}}}\varphi_i\), turning these into quadratic
products in the functions \(\varphi_{i},\phi_{i,j}\), 
hence one obtains a (quadratic and local)
non-linear system of elliptic second order equations with coefficients
analytic away from the positions of the nuclei. 
The result now follows from the results
cited above \cite{Kato-paper, Morrey-book}. (In fact, this
argument extends to solutions of the more general multiconfiguration
self-consistent field equations, see \cite{Friesecke,Lewin}.)

This idea cannot readily be extended to our case. The operator \(E(\bp)\)
is a pseudodifferential operator of first order, so when re-writing
the Hartree-Fock
equations as described above, one obtains a system of
pseudodifferential equations. This system is, as before, of second
(differential) order in the auxiliary functions
\(\phi_{i,j}\), but only of first (pseudodifferential) order in the
original functions \(\varphi_{i}\). Hence, the leading (second) order
matrix is singular elliptic. Hence (even if we ignore the fact that
the square root is non-local) the above argument does not apply.

To summarize, our approach is as follows. We invert the kinetic energy
in the equation for the orbitals thereby obtaining an integral
equation to which we apply successive differentiations. The
localization argument of Lemma~\ref{Edgardo} together with the
smoothing estimates on \(\phi E(\bp)^{-1}D^{\beta}\chi\)
handle the non-locality of this equation. By working in \(L^p\) for
suitably large \(p\) one can afford the necessary loss of derivatives
from using Sobolev inequalities when treating the non-linear terms.


\section{Proof of analyticity} 
In order to prove that the $\varphi_{i}$'s are real analytic
in $\R^3\setminus \{0\}$ it is sufficient \cite[Proposition
2.2.10]{Krantz} to prove that for every
$\bx_{0} \in \R^3\setminus\{0\}$ there exists an open set $U \subseteq
\R^3\setminus\{0\}$ containing \(\bx_{0}\), and constants
\(\mathcal{C}, \mathcal{R}>0\), such that 
\begin{align}\label{eq:aim-est}
  |\partial^{\beta}\varphi_{i}(\bx)|\le
  \mathcal{C}\,\frac{\beta!}{\mathcal{R}^{|\beta|}}
  \ \text{ for all } \bx\in U \text{ and all }\beta\in\N_{0}^{3}\,.
\end{align}

Let $\bx_{0} \in \mathbb{R}^3 \setminus\{0\}$, and let $\omega$ be
the ball \(B_{R}(\bx_0)\) with center $\bx_0$ and radius
$R:=\min\{1,|\bx_0|/4\}$.
For $\delta>0$ we denote by $\omega_{\delta}$ the set of points in
$\omega$ at distance larger than $\delta$ from $\partial \omega$,
i.e., 
\begin{equation}\label{omegadelta}
  \omega_{\delta}:= \{ \bx \in \omega\,|\, d(\bx,\partial \omega)>\delta\}\,.
\end{equation}
By our choice of $\omega$ we have
$\omega_{\delta}=B_{R-\delta}(\bx_{0})$. Therefore
\(\omega_{\delta}=\emptyset\) for \(\delta\ge R\). In particular,  by
our choice of \(R\), 
\begin{align}\label{eq:omega empty}
  \omega_{\delta}=\emptyset\quad \text{for}\quad\delta\ge 1\,.
\end{align}
 For \(\Omega\subseteq\R^{n}\) and \(p\ge1\) we let \(L^{p}(\Omega)\) denote the
 usual \(L^p\)-space with norm
 \(\|f\|_{L^p(\Omega)}=\big(\int_{\Omega}|f(\bx)|^p\,d\bx\big)^{1/p}\). We write
 \(\|f\|_{p}\equiv \|f\|_{L^p(\R^3)}\). 
In the following we equip the Sobolev space $W^{m,p}(\Omega)$,
\(\Omega\subseteq\R^{n}\), $m \in 
\mathbb{N}$ and $p \in [1,\infty)$, with the norm
\begin{equation}\label{def:Sob-norm}
  \| u \|_{W^{m,p}(\Omega)} := \sum_{|\sigma| \leq m} \| D^{\sigma} u
  \|_{L^{p}(\Omega)}\,.
\end{equation}

Theorem~\ref{HF2} follows from the following proposition.

\begin{prop}\label{lemmaestimate}
Let $Z\alpha <2/\pi $, and let $N\ge2$ be a positive integer such
that $N<Z+1$. 
Let \({\boldsymbol{\varphi}}=\{\varphi_1,\ldots,\varphi_N\}\),
 \(\varphi_{i}\in H^{1/2}(\R^{3})\), \(i=1,\ldots,N\), 
 be solutions to the pseudorelativistic Hartree-Fock equations in
 \eqref{eq:HF}. Let \(\bx_{0}\in\R^{3}\setminus\{0\}\),
 \(R=\min\{1,|\bx_{0}|/4\}\), 
and \(\omega=B_{R}(\bx_{0})\). Define
   \(\omega_{\delta}=B_{R-\delta}(\bx_{0})\) for \(\delta>0\).

Then for all \(p\ge 5\) there exist
constants $C,B>1$ such that for all  
$j\in\N$, for 
all $\epsilon >0$ such that \(\epsilon j\le R/2\),
and for all $i \in \{1, \dots,N\}$ we have 
\begin{equation}\label{eq:lemma2}
  \epsilon^{|\beta|} \|D^{\beta} \varphi_{i}\|_{L^{p}(\omega_{\epsilon j})}
   \leq
  C B^{|\beta|} \; \mbox{ for all } \; \beta \in \N_{0}^3 \; \mbox{ with } \;
  |\beta| \leq j\,. 
\end{equation}
\end{prop}
Given Proposition~\ref{lemmaestimate}, the proof that
the $\varphi_{i}$'s are real analytic is standard, using Sobolev
embedding. 
We give the argument here for completeness.  We then give the proof of
Proposition~\ref{lemmaestimate} in the next section. 

Let $U=B_{R/2}(\bx_{0})=\omega_{R/2}\subseteq\omega$.
Using Theorem~\ref{lemmaSobolev} and \eqref{eq:lemma2} we have $\varphi_i \in C(\overline
U)$. Therefore it suffices to prove \eqref{eq:aim-est} for
$|\beta|\geq 1$. 
Fix $i \in \{1, \dots, N\}$ and consider $\beta \in \N_{0}^3 \setminus
\{0\}$ an arbitrary multiindex. Setting
$j=|\beta|$ and $\epsilon=(R/2)/j$ it follows from
Proposition~\ref{lemmaestimate} (since \(\epsilon j=R/2\)) that there
exists constants $C, B>1$ such 
that
\begin{equation}\label{aa1}
  \|D^{\beta} \varphi_i\|_{L^{p}(\omega_{R/2})}
  \leq C
  \Big(\frac{B}{\epsilon}\Big)^{|\beta|}   
  = C \Big(\frac{2B}{R}\Big)^{|\beta|}|\beta|^{|\beta|}\,, 
\end{equation}
with $C, B$ independent of the choice of $\beta$. 
By Theorem~\ref{lemmaSobolev}  (see also Remark~\ref{rem:Morrey})
there exists a constant \(K_{4}=K_{4}(p,\bx_{0})\) such that, 
for all $\beta' \in \N_{0}^3 \setminus \{ 0\}$,
\begin{align*}
  \sup_{\bx\in U} |D^{\beta'} \varphi_{i}(\bx)| & \leq  K_{4}
  \sum_{|\sigma| \leq 1} 
  \|D^{\beta' +\sigma} \varphi_{i}\|_{L^{p}(\omega_{R/2})}
  \\  
  & \leq K_{4} \sum_{|\sigma| \leq 1} C 
  \Big(\frac{2B}{R}\Big)^{|\sigma|+|\beta'|}
  \big(|\sigma|+|\beta'|\big)^{|\sigma|+|\beta'|}\,,
\end{align*}
using \eqref{aa1}. Using that \(R\le1\le B\), that
\(\#\{\sigma\in\N_{0}^3\,|\,|\sigma|=1\}=3\), and that, from
\eqref{eq:Abra-Ste-binom},
\begin{align*}
  \big(1+|\beta'|\big)^{1+|\beta'|}\le 
  \frac{{\rm e}}{\sqrt{2\pi}}\,{\rm e}^{2|\beta'|}\,|\beta'|!\,,
\end{align*}
this implies that for all \(\beta'\in\N_{0}^{3}\setminus\{0\}\),
\begin{align}\label{eq:final-analyticity-Kra}
  \sup_{\bx\in U} |D^{\beta'} \varphi_{i}(\bx)|
  \le \Big(\frac{8{\rm e}K_{4}CB}{\sqrt{2\pi}R}\Big)
  \Big(\frac{2{\rm e}^2B}{R}\Big)^{|\beta'|}|\beta'|!\,.
\end{align}
Since \(|\sigma|!\le 3^{|\sigma|}\sigma!\) for all
\(\sigma\in\N_{0}^{3}\) (see \eqref{eq:bound-sigma-fak} in
Appendix~\ref{Notation} below), this implies that 
\begin{align}\label{eq:final-analyticity-Kra-bis}
  \sup_{\bx\in U} |D^{\beta'} \varphi_{i}(\bx)|
  \le \mathcal{C}\,\frac{\beta'!}{\mathcal{R}^{|\beta'|}}\,,
\end{align}
for some \(\mathcal{C},\mathcal{R}>0\). 
This proves \eqref{eq:aim-est}.
Hence \(\varphi_{i}\) is real analytic in \(\R^3\setminus\{0\}\). This
finishes the proof of  Theorem~\ref{HF2}.

It therefore remains to prove Proposition~\ref{lemmaestimate}.

\begin{remark}\label{rem:constants}
We here give explicit choices for the constants \(C\) and \(B\) in
Proposition~\ref{lemmaestimate}.  

Let  
\begin{align}\label{eq:bound-U-a,b}
  C_{1}:=\max_{1\le a,b\le N}\Big\|\int_{\R^3}
  \frac{|\varphi_{a}(\by)
  \varphi_{b}(\by)|}{|\cdot-\by|}\,d\by\Big\|_{\infty}\,.
\end{align}
Note that by \eqref{eq:Kato} below, this is finite since \(\varphi_{i}\in
 H^{1/2}(\R^{3})\), \(i=1,\ldots,N\). 

Furthermore, let \(A=A(\bx_0)\ge1\) be such that, for all
\(\sigma\in\N_{0}^{3}\), 
\begin{equation}\label{Wanal0-bis}
  \sup_{\bx \in \omega} |D^{\sigma} V(\bx)| \leq A^{|\sigma|+1}
  |\sigma| !\,. 
\end{equation}
The existence of \(A\) follows from the real analyticity in
\(\omega=B_{R}(\bx_0)\) (recall that \(R=\min\{1,|\bx_0|/4\}\))  of 
\(V=Z\alpha|\cdot|^{-1}\) (see e.~g.~\cite[Proposition
2.2.10]{Krantz}). Assume without restriction that \(A\ge
\alpha^{-1}+\max_{1\le i\le N}|\varepsilon_{i}|\).

Let \(K_{1}=K_{1}(p), K_{2}=K_{2}(p)\), and \(K_{3}=K_{3}(p)\) be the
constants in Lemma~\ref{bounded-mult-op}, Corollary~\ref{cor:adams},
and  Corollary~\ref{rem:elliptic}, 
respectively (see Appendices~\ref{app:smoothing} and \ref{app:needed} below).
Then let
\begin{align}
  \label{eq:C-111}
  C_{2}&=\max\big\{K_{1},256\sqrt{2}/\pi\,\big\}\,,
  \\
  \label{eq:const-lem7}
  C_{3} & = \max\big\{4\pi(1+2C_{1}/R^{2})K_{3}, 160\pi
  K_{2}^2K_{3}\big\}\,. 
\end{align}
Choose
\begin{align}\label{rem:choice-C}\nonumber
  C>\max_{i\in\{1,\ldots,N\}}\big\{&1,
  \|\varphi_{i}\|_{W^{1,p}(\omega)},
  \|\varphi_{i}\|_{L^{3p}(B_{2R}(\bx_{0}))},  
  \frac{768}{\pi}|\bx_{0}|^{3(2-p)/(2p)}\|\varphi_{i}\|_{2},
  \\&\quad
  \big[\frac{48\sqrt2}{\pi}A+
  48\sqrt{2}C_{1}\frac{N}{Z\pi}+
  \frac{1536\sqrt2}{\pi^2|\bx_{0}|}
  \big]\|\varphi_{i}\|_{3}
  \big\}\,.
\end{align}
That \(C<\infty\) follows from the smoothness away from \(\bx=0\) of
the \(\varphi_{i}\)'s \cite[Theorem~1 (ii)]{relHF} and the fact that,
since \(\varphi_{i}\in H^{1/2}(\R^{3})\), \(1\le i\le N\), we have
\(\varphi_{i}\in L^{3}(\R^{3})\), \(1\le i\le N\), by Sobolev's
inequality.  
Then choose
\begin{align}\label{rem:choice-B}
  B>
  \max\big\{ 48AC_{2}, C_{*}, \frac{16}{|\bx_{0}|}, 4C_{1}^2,
  (160C^2K_{2}C_{3})^2, (24NC_{2}/Z)^2, 16K_{3}\big\}\,,
\end{align}
where \(C_{*}\) is the constant (related to a smooth partition of
unity) introduced in \eqref{eq:est-der-loc}.
In particular, \(B>48\).
We will prove Proposition~\ref{lemmaestimate} with these choices of
\(C\) and \(B\). 
\end{remark}


\section{Proof of the main estimate}\label{sec:main estimate}
We first make \eqref{eq:HF-functional} more
precise, thereby also explaining the choice of \(\mathcal{M}_{N}\)
in \eqref{eq:HF-constraints}.
 By Kato's inequality \cite[(5.33) p.\
307]{Kato},
\begin{align}\label{eq:Kato}
  \int_{\rz^3}\frac{|f(\bx)|^2}{|\bx|}\,d\bx
  \le \frac{\pi}{2}\int_{\rz^3}|\bp||\hat{f}(\bp)|^2\,d\bp\  \text{
   for }\ 
   f\in H^{1/2}(\rz^3)
\end{align}
(where \(\hat{f}(\bp)=(2\pi)^{-3/2}\int_{\R^3}{\rm e}^{-{\rm
    i}\bx\cdot\bp}f(\bx)\,d\bx\) denotes the Fourier transform of
\(f\)), and the KLMN theorem \cite[Theorem X.17]{RS2} 
the operator \(h_0\) given as
\begin{align}\label{def:h_0}
  h_0=T(-{\rm i}\nabla)-V
\end{align} 
is well-defined on \(H^{1/2}(\R^{3})\) 
(and bounded below by \(-\,\alpha^{-1}\))
as a form sum when
\(Z\alpha<2/\pi\), that is, 
\begin{align}\label{1p}\nonumber
  (u,h_0v)=(E(\bp)^{1/2}u,E(\bp)^{1/2}v)-\alpha^{-1}(u,&v)
  -(V^{1/2}u,V^{1/2}v)
  \\& \text{ for }\  u,v \in
  H^{1/2}(\mathbb{R}^3)\,.
\end{align}
By abuse of notation, we write
 \(E(\bp)\) for the (strictly positive) operator \(E(-{\rm
   i}\nabla)=\sqrt{-\Delta+\alpha^{-2}}\).
For \((\varphi_{1},\ldots,\varphi_{N})\in\mathcal{M}_{N}\), the
function \(R_{\boldsymbol{\varphi}}\) given in \eqref{Rgamma} belongs
to \(L^{\infty}(\R^{3})\) (using Kato's inequality above), and the operator
\(K_{\boldsymbol{\varphi}}\) given in \eqref{Kgamma} is
Hilbert-Schmidt (see \cite[Lemma~2]{relHF}). As a consequence, when
\(Z\alpha<2/\pi\), the
operator \(h_{\boldsymbol{\varphi}}\) in \eqref{eq:HF-operator} 
is a well-defined self-adjoint 
operator 
with quadratic form domain $H^{1/2}(\mathbb{R}^3)$ such
that 
\begin{equation}\label{hgamma}
  (u, h_{{\boldsymbol{\varphi}}} v)= 
  (u,h_{0} v)+
  \alpha (u,R_{{\boldsymbol{\varphi}}}v)-\alpha (u,K_{{\boldsymbol{\varphi}}}v)
  \,\text{ for }\,u,v\in
  H^{1/2}(\mathbb{R}^3)\,. 
\end{equation}
Since \((u,
R_{{\boldsymbol{\varphi}}}u)-(u,K_{{\boldsymbol{\varphi}}}u)\ge0\) for
any \(u\in L^{2}(\R^{3})\), also 
\(h_{{\boldsymbol{\varphi}}}\) is bounded from below by
\(-\,\alpha^{-1}\). 

Then, for \((u_1,\ldots,u_N)\in\mathcal{M}_{N}\), the precise version
of \eqref{eq:HF-functional} becomes
\begin{align}\label{eq:HF-functional-bis}\nonumber
  \mathcal{E}^{\rm HF}&(u_1,\ldots,u_N)
   \\\nonumber&=\sum_{j=1}^{N} \alpha^{-1}(u_{j},h_0u_{j})
   +\frac{1}{2}\sum_{1\le i,j\le N}\int_{\R^3}\int_{\R^3}
   \frac{|u_{i}(\bx)|^2|u_{j}(\by)|^2}{|\bx-\by|}
   \,d\bx d\by
   \\&\qquad\qquad\quad 
    -\frac{1}{2}\sum_{1\le i,j\le N}\int_{\R^3}\int_{\R^3}
   \frac{\overline{u_{j}(\bx)}u_{i}(\bx)\overline{u_{i}(\by)}u_{j}(\by)}{|\bx-\by|} 
   \,d\bx d\by\,.
\end{align}
The considerations on \(R_{{\boldsymbol{\varphi}}}\) and
  \(K_{{\boldsymbol{\varphi}}}\) above 
imply that also the non-linear
terms in \eqref{eq:HF-functional-bis} are finite for \(u_{i}\in
H^{1/2}(\R^{3})\), \(1\le i\le N\).  

If \((\varphi_1,\ldots,\varphi_N)\in\mathcal{M}_{N}\) is a critical
point of \(\mathcal{E}^{\rm HF}\) in \eqref{eq:HF-functional-bis},
then 
\({\boldsymbol{\varphi}}=\{\varphi_1,\ldots,\varphi_N\}\) satisfies
the self-consistent HF-equations \eqref{eq:HF-operator eigen} with
the operator \(h_{\boldsymbol{\varphi}}\) defined above.

Note that \(E(\bp)\) is a bounded operator
from \(H^{1/2}(\R^{3})\) to \(H^{-1/2}(\R^{3})\),  
and recall that \eqref{eq:Kato} shows that \(V\) also defines a bounded
operator from \(H^{1/2}(\R^{3})\) to \(H^{-1/2}(\R^{3})\) (for any \(Z\alpha\)).
 As noted above, both \(R_{\boldsymbol{\varphi}}\) and
 \(K_{\boldsymbol{\varphi}}\) are bounded operators on
 \(L^{2}(\R^{3})\) when
 \((\varphi_1,\ldots,\varphi_N)\in\mathcal{M}_{N}\). 
In particular, this shows that if
\((\varphi_1,\ldots,\varphi_N)\in\mathcal{M}_{N}\) solves
\eqref{eq:HF-operator eigen}, then
\begin{align}\label{eq:in H-1/2}
   E(\bp)\varphi_{i}-\alpha^{-1}\varphi_{i}-V\varphi_{i}+\alpha
   R_{\boldsymbol{\varphi}}\varphi_{i} - \alpha
   K_{\boldsymbol{\varphi}}\varphi_{i} = \varepsilon_{i}\varphi_{i}
   \ , \ 
  1\le i\le N\,,
\end{align}
hold as equations in \(H^{-1/2}(\R^{3})\).
Using that
 \(E(\bp)^{-1}\)  is a bounded operator 
 from \(H^{-1/2}(\R^{3})\) to \(H^{1/2}(\R^{3})\), this implies that,
 as  equalities in \(H^{1/2}(\R^{3})\) (and therefore, in particular, in
 \(L^{2}(\R^{3})\)), 
 \begin{align}\label{eq:use-equation}\nonumber
   \varphi_{i} &= E(\bp)^{-1}  V \varphi_i
     -\alpha E(\bp)^{-1}  R_{{\boldsymbol{\varphi}}} \varphi_i
    \\&\ +\alpha E(\bp)^{-1}  K_{{\boldsymbol{\varphi}}}
      \varphi_i
     +(\alpha^{-1}+\varepsilon_i) E(\bp)^{-1}
       \varphi_i
  \ , \quad
  1\le i\le N\,,
 \end{align}
\begin{pf*}{Proof of Proposition~\ref{lemmaestimate}}
The proof of Proposition~\ref{lemmaestimate} is by induction on
\(j\in\N_{0}\). More precisely:
\begin{defn}\label{def:ihy}
For \(p\ge 1\) and \(j\in\N_{0}\), let \(\mathcal{P}(p,j)\) be the
statement:

For all
$\epsilon >0$ with \(\epsilon j\le R/2\), and all $i
\in \{1, \dots,N\}$ we have  
\begin{equation}\label{eq:lemma2-ihy}
  \epsilon^{|\beta|} \|D^{\beta} \varphi_{i}\|_{L^{p}(\omega_{\epsilon j})}
  \leq
  C\, B^{|\beta|} \; \mbox{ for all } \; \beta \in \N_{0}^3 \; \mbox{ with } \;
  |\beta|\le j\,,
\end{equation}
with \(C,B>1\) the constants in Remark~\ref{rem:constants}.
\end{defn}

Then Proposition~\ref{lemmaestimate} is equivalent to the statement:
For all \(p\ge5\),  \(\mathcal{P}(p,j)\) holds for all \(j\in\N_{0}\). 
This is the statement we will prove by induction on \(j\in\N_{0}\).

\ %

{\bf Induction start:} For convenience, we prove the induction start
for both \(j=0\) and \(j=1\).

Note that \(\mathcal{P}(p,0)\)
trivially holds since (see Remark~\ref{rem:constants})
\begin{align}\label{eq:cond-C-one}
  C=C(p)>\max_{1\le i\le N}\|\varphi_{i}\|_{L^{p}(\omega)}\,.
\end{align}
Also \(\mathcal{P}(p,1)\) holds by the choice of \(C\), since
\begin{align}\label{eq:cond-C-two}
  C=C(p)>\max_{\stackrel{1\le i\le N,}
  {\nu\in\{1,2,3\}}}\|D_{\nu}\varphi_{i}\|_{L^{p}(\omega)}\,.
\end{align}
Namely, since \(\omega_{\epsilon}\subseteq\omega\),
\eqref{eq:lemma2-ihy} holds for \(|\beta|=0\)  (and {\it all}
\(\epsilon>0\)) using
\eqref{eq:cond-C-one}. For
\(\beta\in\N_{0}\) with \(|\beta|=1=j\) (i.e., \(\beta=e_{\nu}\) for
some \(\nu\in\{1,2,3\}\)), and all \(\epsilon>0\) with \(\epsilon
=\epsilon j\le R/2<1\),  
\begin{align}\label{eq:ind-statr-for-beta=1}\nonumber
  \epsilon^{|\beta|}
  \|D^{\beta}\varphi_{i}\|_{L^{p}(\omega_{\epsilon j})}
  &=\epsilon\|D_{\nu}\varphi_{i}\|_{L^{p}(\omega_{\epsilon})}
  \le\|D_{\nu}\varphi_{i}\|_{L^{p}(\omega)}
  \\&\le C\le C B=C B^{|\beta|}\,.
\end{align}
Here we again used that \(\omega_{\epsilon}\subseteq\omega\),
\eqref{eq:cond-C-two}, and that \(B>1\) (see 
Remark~\ref{rem:constants}). 

We move on to the induction step.\\

\noindent{\bf Induction hypothesis:}
\begin{align}\label{eq:ihy}
\text{Let \(p\ge5\) and \(j\in\N_{0}\), \(j\ge 1\). Then
\(\mathcal{P}(p,\tilde{j})\) holds for all \(\tilde{j}\le j\).}
\end{align}
 
We now prove that \(\mathcal{P}(p,j+1)\) holds. Note that to prove
this, it suffices to study \(\beta\in\N_0^3\)
with \(|\beta|=j+1\). Namely, assume \(\epsilon>0\) is such that
\(\epsilon(j+1)\le R/2\) and let \(\beta\in\N_{0}^{3}\) with
\(|\beta|<j+1\). Then \(|\beta|\le j\) and \(\epsilon j\le R/2\) so,
by the definition of \(\omega_{\delta}\) and the induction hypothesis, 
\begin{align}\label{eq:enough-top-beta}
  \epsilon^{|\beta|}
  \|D^{\beta}\varphi_{i}\|_{L^{p}(\omega_{\epsilon(j+1)})}
  \le\epsilon^{|\beta|}
  \|D^{\beta}\varphi_{i}\|_{L^{p}(\omega_{\epsilon j})}
  \le C B^{|\beta|}\,.
\end{align}
It therefore remains to prove that
\begin{align}\label{est:to-prove}\nonumber
  \epsilon^{|\beta|} 
  \|D^{\beta} \varphi_{i}\|_{L^{p}(\omega_{\epsilon(j+1)})}
  \leq
  C\, &B^{|\beta|} \quad 
  \mbox{ for all } \; \epsilon>0 \; \mbox{ with }
  \; \epsilon(j+1)\le R/2 
  \\
  &\mbox{ and all }
  \beta \in \N_{0}^3 \; \mbox{ with } \;
  |\beta|=j+1\,.
\end{align}

\begin{remark}\label{ihy}
To use the induction hypothesis in its entire strength, it is
convenient to write, for  \(\ell>0\), \(\epsilon>0\) such that
\(\epsilon \ell\le 
R/2\), and
\(\sigma\in\N_0^3\) with \(0<|\sigma|\le j\),
\begin{equation*}
  \|D^{\sigma}\varphi_{i}\|_{L^{p}(\omega_{\epsilon \ell})}
  = 
  \|D^{\sigma}\varphi_{i}\|_{L^{p}(\omega_{\tilde{\epsilon}\tilde{j}})}
  \quad\text{with}\quad
  \tilde{\epsilon}=\frac{\epsilon \ell}{|\sigma|},\ \tilde{j}=|\sigma|
  \,, 
\end{equation*}
so that, by the induction hypothesis (applied on the term with
$\tilde{\epsilon}$ and \(\tilde{j}\)) we get that
\begin{align}\label{eq:ihy-new}
  \|D^{\sigma}\varphi_{i}\|_{L^{p}(\omega_{\epsilon\ell})}
  \leq C
  \Big(\frac{B}{\tilde{\epsilon}}\Big)^{|\sigma|}
  = C 
  \Big( \frac{|\sigma|}{\ell} \Big)^{|\sigma|}
  \Big(\frac{B}{\epsilon}\Big)^{|\sigma|}\,.  
\end{align}
Compare this with \eqref{eq:lemma2-ihy}.
With the convention that \(0^0=1\), \eqref{eq:ihy-new} also holds for \(|\sigma|=0\). 
\end{remark}
We choose a function  $\Phi$ (depending on \(j\)) satisfying 
\begin{equation}\label{def:Phi}
  \Phi \in C^{\infty}_{0}(\omega_{\epsilon (j+3/4)})\,,\quad
  0\le\Phi\le1\,,\quad
  \mbox{
  with }\; \Phi \equiv 1 \; \mbox{ on }\; \omega_{\epsilon(j+1)}\,. 
\end{equation}
Then
\begin{align}\label{dual}
   \|D^{\beta}\varphi_i\|_{L^{p}(\omega_{\epsilon(j+1)})}
   \le \|\Phi D^{\beta}\varphi_i\|_{p} \,.
\end{align}
The estimate \eqref{est:to-prove}--and hence, by induction, the proof of
Proposition~\ref{lemmaestimate}---now follows from the equations
\eqref{eq:use-equation} for the 
\(\varphi_{i}\)'s, 
\eqref{dual}
and the following two lemmas.  

\begin{lemma} \label{lem:V}
Assume \eqref{eq:ihy} (the induction
  hypothesis) holds. Let \(\Phi\) 
be as in \eqref{def:Phi}.
Then for all \(i\in\{1,\ldots,N\}\), 
all $\epsilon >0$ with \(\epsilon(j+1)\le R/2\), and all $\beta \in
\N_{0}^3$ with $|\beta|=j+1$, both \(\Phi D^{\beta} E(\bp)^{-1}  V
\varphi_i\) and \(\Phi D^{\beta} E(\bp)^{-1}  \varphi_i\) belong to
\(L^{p}(\R^{3})\), and
\begin{align}\label{est-V}
   \|\Phi D^{\beta} E(\bp)^{-1}  V \varphi_i\|_{p}
   &\leq \frac{C}{4}\Big(\frac{B}{\epsilon}\Big)^{|\beta|}\,,\\ 
   \label{est-V-bis}
  \|(\alpha^{-1}+\varepsilon_{i})\Phi D^{\beta} E(\bp)^{-1}  \varphi_i\|_{p}
   &\leq \frac{C}{4}\Big(\frac{B}{\epsilon}\Big)^{|\beta|}\,, 
\end{align}
where $C, B>1\) are the constants in \eqref{eq:lemma2-ihy} (see also
Remark~\ref{rem:constants}). 
\end{lemma}
\begin{lemma} \label{lem:R-K}
Assume \eqref{eq:ihy} (the induction hypothesis) holds. Let \(\Phi\)
be as in \eqref{def:Phi}.
Then for all \(i\in\{1,\ldots,N\}\), 
all $\epsilon >0$ with \(\epsilon(j+1)\le R/2\), and all $\beta \in
\N_{0}^3$ with $|\beta|=j+1$, both \(\Phi D^{\beta} E(\bp)^{-1}  R_{{\boldsymbol{\varphi}}} 
  \varphi_i\) and \(\Phi D^{\beta}  E(\bp)^{-1}  K_{{\boldsymbol{\varphi}}}
   \varphi_i\) belong to
\(L^{p}(\R^{3})\), and
\begin{align*}
  \|\alpha\,\Phi D^{\beta} E(\bp)^{-1}  R_{{\boldsymbol{\varphi}}} 
  \varphi_i\|_{p}
  &\leq \frac{C}{4}\Big(\frac{B}{\epsilon}\Big)^{|\beta|}\,, \\
  \|\alpha\,\Phi D^{\beta}  E(\bp)^{-1}  K_{{\boldsymbol{\varphi}}}
  \varphi_i\|_{p}
  &\leq \frac{C}{4}\Big(\frac{B}{\epsilon}\Big)^{|\beta|}\,,  
\end{align*}
where $C, B>1\) are the constants in \eqref{eq:lemma2-ihy} (see also
Remark~\ref{rem:constants}). 
\end{lemma}
\begin{remark}
  For $a,b \in \{1, \dots,N\}$, let $U_{a,b}$ denote the function 
 \begin{equation}\label{eq:def-U-a,b}
   U_{a,b}(\bx) = \int_{\R^3} \frac{\varphi_{a}(\by)
     \overline{\varphi_{b}(\by)}}{|\bx-\by|}\,d\by\,,\ \bx\in\R^3\,.
 \end{equation}
In particular, \(\|U_{a,b}\|_{\infty}\le C_{1}\) for all
\(a,b\in\{1,\ldots,N\}\) (see \eqref{eq:bound-U-a,b}). Note that (see
\eqref{Rgamma} and \eqref{Kgamma}) 
\begin{align}\label{eq:K-R-U}
  R_{{\boldsymbol{\varphi}}} \varphi_i
  =\sum_{\ell=1}^{N}U_{\ell,\ell} \varphi_{i}\,,\quad
  K_{{\boldsymbol{\varphi}}} \varphi_i
  =\sum_{\ell=1}^{N}U_{i,\ell} \varphi_{\ell}\,.
\end{align}
Hence Lemma~\ref{lem:R-K} follows from the following lemma and the
fact that \(Z\alpha<2/\pi<1\).
\end{remark}

\begin{lemma} \label{lem:U-phi}
 Assume \eqref{eq:ihy} (the induction hypothesis) holds. Let \(\Phi\)
be as in \eqref{def:Phi}.
For $a,b\in \{1, \dots, N\}$, let \( U_{a,b}\) be given
by \eqref{eq:def-U-a,b}. 
Then for all $a,b ,i \in \{1, \dots, N\}$,
all $\epsilon >0$ with \(\epsilon(j+1)\le R/2\), 
 and all 
 \(\beta \in \N_{0}^3\) with \(|\beta|=j+1\), 
 \(\Phi D^{\beta} E(\bp)^{-1}U_{a,b} \varphi_i\)  belong to
 \(L^{p}(\R^{3})\), and
\begin{align}\label{est:U-lemma}
 \|\Phi D^{\beta} E(\bp)^{-1}U_{a,b} \varphi_i\|_{p}
  \leq \frac{CZ}{4N}\Big(\frac{B}{\epsilon}\Big)^{|\beta|}\,, 
\end{align}
where $C, B>1\) are the constants in \eqref{eq:lemma2-ihy} (see also
Remark~\ref{rem:constants}). 
\end{lemma}

It therefore remains to prove Lemmas~\ref{lem:V} and
\ref{lem:U-phi}. This will be done in the two
following sections.
\end{pf*}
\section{Proof of Lemma~\ref{lem:V}}
We prove Lemma~\ref{lem:V} by proving \eqref{est-V} and \eqref{est-V-bis}
separately. 

\begin{pf*}{Proof of \eqref{est-V}}
 Let $\sigma
\in \N_{0}^3$ and $\nu\in \{1,2,3\}$ be such that $\beta=\sigma
+e_{\nu}$, so that $D^{\beta}=D_{\nu}D^{\sigma}$. Notice that 
$|\sigma|=j$.
Choose localization functions
$\{\chi_{k}\}_{k=0}^{j}$ and $\{\eta_{k}\}_{k=0}^{j}$ as in
Appendix~\ref{localization} below. 
Since \(V\varphi_{i}\in H^{-1/2}(\R^{3})\), and \(E(\bp)^{-1}\) maps
\(H^{s}(\R^{3})\) to \(H^{s+1}(\R^{3})\) for all \(s\in\R\), 
Lemma~\ref{Edgardo} 
(with $\ell=j$) implies that
\begin{align}\notag
  \Phi D^{\beta} E(\bp)^{-1}[V \varphi_{i}] 
  &=\sum_{k=0}^{j} 
  \Phi D_{\nu}E(\bp)^{-1}D^{\beta_{k}}\chi_{k} D^{\sigma-\beta_{k}}
  [V\varphi_{i}] \\  
  &{\ }+ \sum_{k=0}^{j-1} 
    \Phi D_{\nu}E(\bp)^{-1}D^{\beta_{k}}[\eta_{k},D^{\mu_{k}}] 
  D^{\sigma-\beta_{k+1}} [V
  \varphi_{i}] \notag \\ \label{f2} 
  &{\ }+  
  \Phi D_{\nu}E(\bp)^{-1}D^{\sigma}[\eta_{j}V\varphi_{i}]\,,
\end{align}
as an identity in \(H^{-|\beta|+1/2}(\R^{3})\) (we have also used that
\(E(\bp)^{-1}\) commutes with derivatives on any \(H^{s}(\R^{3})\)). 
Here, \([\,\cdot\,,\,\cdot\,]\) denotes the commutator. Also,
\(|\beta_k|=k\), \(|\mu_k|=1\), and \(0\le\eta_k,\chi_k\le1\). (For
the support properties of \(\eta_k,\chi_k\), see the mentioned appendix.)
We will prove that each term on the right side of \eqref{f2}
belong to \(L^{p}(\R^{3})\), and bound their norms.
The proof of \eqref{est-V} will follow by summing these bounds.\\

\noindent\textit{The first sum in \eqref{f2}.} 
Let $\theta_{k}$ be the characteristic function of
the support of $\chi_{k}$ (which is contained in \(\omega\)). 
Since \(V\) is smooth on the closure of
\(\omega\) it follows from the induction hypothesis that the
\(D^{\sigma-\beta_{k}} [V\varphi_{i}]\)'s belong to \(L^{p}(\omega')\)
for any \(\omega'\subset\subset\omega\).
Also, the operator \(\Phi D_{\nu}E(\bp)^{-1}D^{\beta_{k}}\chi_{k}\)
is bounded on \(L^{p}(\R^{3})\) (as we will observe below).
Therefore we can estimate, for \(k\in\{0,\ldots,j\}\),
\begin{align}\label{f3}\notag 
  \|\Phi D_{\nu}E(\bp)^{-1}D^{\beta_{k}}&\chi_{k}
  D^{\sigma-\beta_{k}} [V \varphi_{i}]
  \|_{p}\\\nonumber
   & =\| (\Phi E(\bp)^{-1}D_{\nu}D^{\beta_{k}} \chi_{k})
   \theta_k D^{\sigma-\beta_{k}} [V \varphi_{i}]
   \|_{p}
  \\&\leq  \| \Phi E(\bp)^{-1}D_{\nu}D^{\beta_{k}} \chi_{k}  
   \|_{\mathcal{B}_{p}}\, \| \theta_{k}  D^{\sigma-\beta_{k}} [
   V \varphi_{i}] \|_{p}\,.
\end{align}
Here, \(\|\cdot\|_{\mathcal{B}_{p}}\) is the operator norm on
\(\mathcal{B}_{p}:=B(L^{p}(\R^3))\), the bounded operators 
  on \(L^{p}(\R^3)\).

For \(k=0\), the first factor on the right side of \eqref{f3} can be
estimated using Lemma~\ref{bounded-mult-op} 
(since \(|\beta_0|=0\)).
This way, since \(\|\chi_{0}\|_{\infty}=\|\Phi\|_{\infty}=1\), 
\begin{align}\label{eq:est-smooth-factor-k=0}
  \|\Phi E(\bp)^{-1}D_{\nu}\chi_{0}  
  \|_{\mathcal{B}_{p}} \le K_{1}\,,
\end{align}  
with \(K_{1}=K_{1}(p)\) the constant in \eqref{norm-mult-op}. 

For \(k>0\), the first factor on the right side of \eqref{f3} can be 
estimated using \eqref{eq:smoothing-est-Lp-bis}
in Lemma~\ref{normsmooth-Lp} 
(with \(\mathfrak{r}=1\), \(\mathfrak{q}^{*}=\mathfrak{p}=p\)).
Since 
\begin{align*}
 \dist(\supp\,\chi_{k},\supp\,\Phi)\ge\epsilon 
  (k-1+1/4)
\end{align*}
 and \(\|\chi_k\|_{\infty}=\|\Phi\|_{\infty}=1\), this gives 
(since \((\beta_k+e_{\nu})!\le(|\beta_k|+1)!=(k+1)!\)) that
\begin{align}\label{eq:est-smooth-factor-k>0}
  \nonumber
  \| \Phi E(\bp)^{-1} D_{\nu}D^{\beta_{k}} \chi_{k}
  \|_{\mathcal{B}_{p}}
  &\le \frac{32\sqrt{2}}{\pi}
  \frac{(k+1)!}{k}\Big(\frac{8}{\epsilon 
  (k-1+1/4)}\Big)^{k}
  \\&\le \frac{256\sqrt{2}}{\pi}
  \Big(\frac{8}{\epsilon}\Big)^{k}\,.
\end{align} 

 It follows from \eqref{eq:est-smooth-factor-k=0} and
 \eqref{eq:est-smooth-factor-k>0} that, for all \(k\in\{0,\ldots,j\}\), 
\(\nu\in\{1,2,3\}\),
\begin{align}\label{eq:est-smooth-factor-BIS}
  \| \Phi E(\bp)^{-1} D_{\nu}D^{\beta_k}\chi_{k}
   \|_{\mathcal{B}_{p}} \le C_{2}\Big(\frac{8}{\epsilon}\Big)^{k}\,,
\end{align} 
with \(C_{2}\) as defined in \eqref{eq:C-111}.

It remains to estimate the second factor in \eqref{f3}. 
Recall the definition of the constant \(A\) in \eqref{Wanal0-bis}.
It follows from \eqref{Wanal0-bis} and \eqref{eq:omega empty} that, for
all \(\epsilon>0\), $\ell\in 
\N_{0}$, and \(\sigma\in\N_{0}^{3}\),  
\begin{equation}\label{Wanal}
  \epsilon^{|\sigma|} \sup_{\bx \in \omega_{\epsilon \ell}}
  |D^{\sigma} V(\bx)| \leq A^{|\sigma|+1} |\sigma| ! \;  {\ell}^{-|\sigma|}\,,   
\end{equation}
with \(\omega_{\epsilon \ell}\subseteq\omega\) as in defined in \eqref{omegadelta}.

For $k=j$, since $\beta_{j}=\sigma$, we find, by
\eqref{Wanal} and the choice of \(C\) (see
Remark~\ref{rem:constants}), that
\begin{equation}\label{f4}
  \| \theta_{j}V \varphi_{i}
  \|_{p} 
  \leq \|V\|_{L^{\infty}(\omega)}
  \|\varphi_{i}\|_{L^{p}(\omega)}
  \le CA\,. 
\end{equation}

The estimate for $k \in \{0, \dots, j-1\}$ is a bit more involved. 
We get, by Leibniz's rule, that 
\begin{align}\label{f5}\nonumber
   \| \theta_{k} D^{\sigma-\beta_{k}}
  &[V \varphi_{i}] \|_{p}\\
  &\leq \sum_{\mu \leq \sigma-\beta_{k}} 
  \binom{\sigma-\beta_k}{\mu}
  \| \theta_{k} D^{\mu} V\|_{\infty} 
  \,
  \|\theta_{k}D^{\sigma-\beta_{k}-\mu} 
  \varphi_{i}\|_{p}\,.
\end{align}
Now, \(\supp\,\theta_{k}=\supp\,\chi_{k}\subseteq \omega_{\epsilon(j-k+1/4)}\), so
by \eqref{Wanal}, for all \(\mu\le \sigma-\beta_k\),  
\begin{align}\label{eq:est-infty-norm-V}
   \| \theta_{k} D^{\mu} V\|_{\infty} 
   \le \sup_{\bx\in\omega_{\epsilon(j-k+1/4)}}
   |D^{\mu} V(\bx)|
   \le \epsilon^{-|\mu|}A^{|\mu|+1}|\mu|!(j-k)^{-|\mu|}\,.
\end{align}
By the induction
hypothesis  (in 
the form discussed in Remark~\ref{ihy}),
\begin{align}\label{eq:est-phi-in-V-term}\nonumber
   \|\theta_{k}D^{\sigma-\beta_{k}-\mu} 
  \varphi_{i}\|_{p}
  &\le \|D^{\sigma-\beta_{k}-\mu} 
  \varphi_{i}\|_{L^{p}(\omega_{\epsilon(j-k)})}
   \\&\le C\Big(\frac{|\sigma-\beta_k-\mu|}{j-k}\Big)^{|\sigma-\beta_k-\mu|}
  \Big(\frac{B}{\epsilon}\Big)^{|\sigma-\beta_k-\mu|}\,.
\end{align}
It follows from \eqref{f5}, \eqref{eq:est-infty-norm-V},
and \eqref {eq:est-phi-in-V-term}
that (using that \(|\sigma|=j, |\beta_k|=k\), and
\eqref{eq:multiNom},  summing over \(m=|\mu|\)) 
\begin{align}\label{eq:first-sum-V-1-bis}\nonumber
  &\| \theta_{k}D^{\sigma-\beta_{k}}[V
  \varphi_{i}]\|_{p}
  \\&\le 
  C A \Big(\frac{B}{\epsilon}\Big)^{j-k}
   \sum_{m=0}^{j-k}
  \binom{j-k}{m} 
  \frac{m!(j-k-m)^{j-k-m}}{(j-k)^{j-k}}
  \Big(\frac{A}{B}\Big)^{m}\,.
\end{align}
Note that, by \eqref{eq:Abra-Ste-binom}, for \(0<m<j-k\),
\begin{align}\label{eg:est-V-fraction}
  \binom{j-k}{m} \frac{m!(j-k-m)^{j-k-m}}{(j-k)^{j-k}}
  \le\frac{{\rm e}^{1/12}\sqrt{j-k}}{\sqrt{j-k-m}\,{\rm e}^{m}}
  \le 1\,.
\end{align}
To see the last inequality, look at the cases \(0<m\le(j-k)/2\) and
\(j-k>m\ge(j-k)/2\) separately.

Hence (since \(B>2A\), see
Remark~\ref{rem:constants}), for any \(k\in\{0,\ldots,j-1\}\), 
\begin{align}\label{eq:first-sum-V-1}
 \| \theta_{k}D^{\sigma-\beta_{k}}[V
  \varphi_{i}]\|_{p}
  &\le C A \Big(\frac{B}{\epsilon}\Big)^{j-k}
  \sum_{m=0}^{j-k}
  \Big(\frac{A}{B}\Big)^{m}
  \le 2C A \Big(\frac{B}{\epsilon}\Big)^{j-k}\,.
\end{align}
Note that, by \eqref{f4}, the same estimate holds true if \(k=j\).

So, from \eqref{f3}, \eqref{eq:est-smooth-factor-BIS}, 
\eqref{eq:first-sum-V-1}, the fact that \(\epsilon\le1\) (since
\(\epsilon(j+1)\le R/2\le1/2\)), and the
choice of \(B\) 
(in particular, \(B>16\); see
Remark~\ref{rem:constants}), it follows that  
\begin{align}\label{eq:final-first-sum-V}\nonumber
   \Big\|&\sum_{k=0}^{j}
  \Phi D_{\nu} E(\bp)^{-1} D^{\beta_{k}} 
  \chi_{k} D^{\sigma-\beta_{k}} [V \varphi_{i}]\Big\|_{p}
  \\&\le
  2CA C_{2}\Big(\frac{B}{\epsilon}\Big)^{j}
  \sum_{k=0}^{j}
  \Big(\frac{8}{B}\Big)^{k}\le
   C(4 A C_{2})\Big(\frac{B}{\epsilon}\Big)^{j}
  \le \frac{C}{12}\Big(\frac{B}{\epsilon}\Big)^{j+1}\,.
\end{align}

\ %

\noindent\textit{The second sum in \eqref{f2}.} 
Note first that
\([\eta_k,D^{\mu_k}]=-(D^{\mu_k}\eta_k)\) (recall that \(|\mu_k|=1\);
see Lem\-ma~\ref{Edgardo}).

Comparing the second sum in \eqref{f2} with the first sum in \eqref{f2}, one sees 
that the second sum is the first one with \(j\) replaced by \(j-1\) and 
\(\chi_k\) replaced by \(-D^{\mu_k}\eta_k\). Having now a derivative on 
the localization functions we have one derivative less falling on the 
term \(V\varphi_i\). More precisely, the operator
\(D^{\sigma-\beta_{k+1}}\) contains
\(|\sigma-\beta_{k+1}|=j-(k+1)=(j-1)-k\) derivatives instead  of
\(|\sigma-\beta_k|=j-k\) in \(D^{\sigma-\beta_k}\). 
Then, to control \(D^{\sigma-\beta_{k+1}} 
[V\varphi_i]\) (with the same method used above for  
\(D^{\sigma-\beta_{k}} [V\varphi_i]\)) we need that 
\(\supp\,D^{\mu_{k}} \eta_k\) is contained in
\(\omega_{\epsilon((j-1)-k+1/4)}\).  
Indeed we have much more: as for \(\chi_k\) we have
\(\supp\,D^{\mu_{k}}  
\eta_k \subseteq \omega_{\epsilon(j-k+1/4)} \subseteq 
\omega_{\epsilon((j-1)-k+1/4)}\).
Finally, 
\(\|D^{\mu_{k}}\eta_k\|_{\infty}\le C_*/\epsilon\), with
\(C_{*}>0\) the constant in \eqref{eq:est-der-loc} in
Appendix~\ref{localization} below.

It follows that the second sum in \eqref{f2} can be
estimated as the first one, up to {\it one} extra
factor of \(C_*/\epsilon\) 
{\it and} up to replacing \(j\) by \(j-1\) in
the estimate \eqref{eq:final-first-sum-V}. Hence, using that
\(\epsilon\le1\), and the choice of \(B\) 
(see Remark~\ref{rem:constants}), we get that
\begin{align}\label{eq:second-sum-total-final-V}\nonumber
  \Big\|\sum_{k=0}^{j-1} 
   &\Phi D_{\nu}E(\bp)^{-1}D^{\beta_{k}}
  [\eta_{k},D^{\mu_{k}}] D^{\sigma-\beta_{k+1}}
  [V \varphi_{i}]\Big\|_{p}
  \\&\le  \frac{C_{*}}{\epsilon}
  C(4A C_{2})\Big(\frac{B}{\epsilon}\Big)^{j-1}
  \le C
  (4AC_{2})
  \Big(\frac{B}{\epsilon}\Big)^{j}
  \le \frac{C}{12}\Big(\frac{B}{\epsilon}\Big)^{j+1}\,.
\end{align}

\ %

\noindent\textit{The last term in \eqref{f2}.} 
It remains to study
\begin{align}\label{the-term}
  \Phi D^{\beta} E(\bp)^{-1}[\eta_{j} V
  \varphi_{i}]\,.
\end{align}
 We split \(V\) in two parts, one supported
 around \(\bx = 0\), and one supported away from \(\bx = 0\), and study
 the two terms separately. We will prove below that this way,
 \(\eta_jV\varphi_{i}\) is actually a function in
 \(L^{1}(\R^{3})+L^{3}(\R^{3})\). Upon using suitable
 operator bounds
 on \(\Phi D^{\beta}E(\bp)^{-1}\chi\) (for some suitable smooth
 \(\chi\)'s), combined with bounds on the norms of 
 the two parts of \(\eta_jV\varphi_{i}\), we will finish the proof.

Let \(\rho=|\bx_0|/4\), and  let \(\theta_\rho\) and
\(\theta_{\rho/2}\) be the characteristic functions 
of the balls \(B_{\rho}(0)\) and \(B_{\rho/2}(0)\), respectively.
Choose \(\widetilde{\chi}_{\rho}\in
C_{0}^{\infty}(\R^{3})\) with \(\supp\,\widetilde{\chi}_{\rho}\subseteq
B_{\rho}(0)\), \(0\le \widetilde{\chi}_{\rho}\le1\), and \(\widetilde{\chi}_{\rho}=1\) on
\(B_{\rho/2}(0)\). 
Note that then 
\begin{align}\label{d-2}
  \dist(\supp\, \Phi,\supp\, \widetilde{\chi}_{\rho})\ge 
  \frac{|\bx_0|}{2}=2\rho\,,
\end{align}
by the choice of \(\omega=B_{R}(\bx_0)\), \(R=\min\{1,|\bx_0|/4\}\),
since \(\supp\,\Phi\subseteq\omega_{\epsilon(j+1)}\subseteq\omega\). 

Now,
\begin{align}\label{split-V}\nonumber
  \Phi D^{\beta} E(\bp)^{-1}[\eta_jV\varphi_i]
  & = \Phi D^{\beta} E(\bp)^{-1}[\eta_jV\widetilde{\chi}_{\rho}\varphi_i]
  \\ &{\quad } + \Phi D^{\beta} E(\bp)^{-1}[\eta_jV(1-\widetilde{\chi}_{\rho})\varphi_i]\,.
\end{align}

For the first term in \eqref{split-V}, we use
Lemma~\ref{normsmooth-Lp}, with \(\mathfrak{p}=1\),
\(\mathfrak{q}=p/(p-1)\), and 
\(\mathfrak{r}=p\). Then
\(\mathfrak{p}, 
\mathfrak{r}\in[1,\infty)\) and \(\mathfrak{q}>1\), 
and \(\mathfrak{q}^{-1}+p^{-1}=1\). 
We get that (recall \eqref{d-2} and that
\(\widetilde{\chi}_{\rho}\theta_{\rho}=\widetilde{\chi}_{\rho}\)), 
\begin{align}\label{first-est}  \nonumber
  \|\Phi & D^{\beta} E(\bp)^{-1}[\eta_{j} V
  \widetilde{\chi}_{\rho}\varphi_{i}]\|_{p}
  \le 
  \|\Phi D^{\beta} E(\bp)^{-1}\widetilde{\chi}_{\rho}\|_{\mathcal{B}_{1,p}}
  \|\eta_{j} V 
  \theta_\rho\varphi_{i}\|_{1}
  \\&
  \le
  \frac{4\sqrt{2}}{\pi}
  \beta!\Big(\frac{8}{2\rho}\Big)^{|\beta|}
  (2\rho)^{3/\mathfrak{r}-2}
   \big(\mathfrak{r}(|\beta|+2)-3
  \big)^{-1/\mathfrak{r}} 
  \|V\theta_\rho\varphi_{i}\|_{1}\,.
\end{align}
Here we used that \(\|\Phi\|_{\infty}=\|\widetilde{\chi}_{\rho}\|_{\infty}=1\)
and that \(\eta_j\equiv1\) where
\(\theta_{\rho}\ne0\). 
Note that \(j+1\le \epsilon^{-1}\) (since, by assumption,
\(\epsilon(j+1)\le R/2\le1/2\)). Therefore,
\begin{align}\label{gain-eps}
  \beta!\le|\beta|!=(j+1)!\le(j+1)^{j+1}\le \epsilon^{-(j+1)}
  =\epsilon^{-|\beta|}\,.
\end{align}
Note furthermore that since \(|\beta|=j+1\ge2\) and
\(\mathfrak{r}\ge1\),  
\begin{align}\label{est:with-r}
  \big(\mathfrak{r}(|\beta|+2)-3\big)^{-1/\mathfrak{r}}
  &\le 1\,,
\end{align}
independently of \(\beta\). It follows that 
\begin{align}\label{sec-est-bis}\nonumber
  \|\Phi D^{\beta} E(\bp)^{-1}&[\eta_{j} V
  \widetilde{\chi}_{\rho}\varphi_{i}]\|_{p}
  \\&\le
  \frac{4\sqrt{2}}{\pi}\Big(\frac{|\bx_{0}|}{2}\Big)^{(3-2p)/p}
  \|V\theta_\rho\varphi_{i}\|_{1}
  \Big(\frac{16/|\bx_{0}|}{\epsilon}\Big)^{|\beta|} \,.
 \end{align}
Using Schwarz's inequality and that \(Z\alpha<2/\pi\),
\begin{align}\label{cond-phi}
  \|V\theta_\rho\varphi_i\|_{1}
  \le \|V\theta_\rho\|_{2}\|\varphi_i\|_{2}
  =Z\alpha\sqrt{|\bx_{0}|\pi}\|\varphi_i\|_{2}
  \le\frac{2}{\sqrt{\pi}}\sqrt{|\bx_{0}|}\|\varphi_i\|_{2}\,.
\end{align}
(Note that \(\|V\theta_\rho\|_{t}<\infty\Leftrightarrow
  t<3\).)
It follows from \eqref{sec-est-bis}, \eqref{cond-phi},
and the choice of \(B\) and \(C\) (see
Remark~\ref{rem:constants}) that
 \begin{align}\label{sec-est-bis-bis}\nonumber
  \|\Phi D^{\beta}&E(\bp)^{-1}[\eta_{j} V
   \widetilde{\chi}_{\rho}
   \varphi_{i}]\|_{p}
   \\&\le 
  \frac{32}{\pi} |\bx_{0}|^{3(2-p)/(2p)}
  \|\varphi_{i}\|_{2}
  \Big(\frac{16/|\bx_{0}|}{\epsilon}\Big)^{|\beta|}
   \le \frac{C}{24}   \Big(\frac{B}{\epsilon}\Big)^{j+1}\,.
\end{align}

We now consider the second term in \eqref{split-V}.
Recall that \(\Phi\) is supported in \(\omega_{\epsilon(j+1)}\) and 
\begin{align}\label{d-bis}
  \dist(\supp\, \Phi,\supp\, \eta_j)\ge \epsilon(j+1/4)\,.
\end{align}
Again, we use Lemma~\ref{normsmooth-Lp}, this time with
\(\mathfrak{p}=3\), \(\mathfrak{q}=p/(p-1)\), and 
\(\mathfrak{r}=3p/(2p+3)\).
Then \(\mathfrak{p}^{-1}+\mathfrak{q}^{-1}+\mathfrak{r}^{-1}=2\),  
\(\mathfrak{p}\in[1,\infty)\), \(\mathfrak{q}>1\),
\(\mathfrak{r}\in [1,3/2)\)  (since \(p>3\)), and
\(\mathfrak{q}^{-1}+p^{-1}=1\). 
This gives that 
\begin{align*}
   \|\Phi D^{\beta} &E(\bp)^{-1}[\eta_{j} V
    (1-\widetilde{\chi}_{\rho})\varphi_{i}]\|_{p}
   \le  \|\Phi D^{\beta}
   E(\bp)^{-1}\eta_{j}\|_{\mathcal{B}_{3,p}}
   \|V(1-\widetilde{\chi}_{\rho})\varphi_i\|_{3}
   \\&\le
   \frac{4\sqrt{2}}{\pi}
   \beta!\Big(\frac{8}{\epsilon(j+1/4)}\Big)^{|\beta|}
   \big(\epsilon(j+1/4)\big)^{3/\mathfrak{r}-2} 
  \big(\mathfrak{r}(|\beta|+2)-3\big)^{-1/\mathfrak{r}}
  \\&\qquad\qquad\qquad\times 
  \|V(1-\widetilde{\chi}_{\rho})\|_{\infty}\|\varphi_i\|_{3}\,.
\end{align*}
As before, we used that
\(\|\Phi\|_{\infty}=\|\eta_{j}\|_{\infty}=1\). 
Note that
\begin{align}\label{est:beta-fak}  
 \beta!\Big(\frac{8}{j+1/4}\Big)^{|\beta|}
  & \le 32^{|\beta|}\frac{|\beta|!}{(j+1)^{|\beta|}}
   =32^{|\beta|}\frac{(j+1)!}{(j+1)^{j+1}}\le 32^{|\beta|}\,.
\end{align}
 Since \(\epsilon(j+1)\le R/2<1\) and \(\mathfrak{r}<3/2\) it follows that
\((\epsilon(j+1/4))^{3/\mathfrak{r}-2}\le 1\). 
Also, by the choice of \(\rho\), the definition of \(V\), and since
\(Z\alpha<2/\pi\), 
\begin{align}\label{eq:V-phi-est}
  \big|((1-\theta_{\rho/2}) V)(\bx)\big|\le \frac{8Z\alpha}{|\bx_0|}
   \le \frac{16}{\pi|\bx_0|}\,, 
  \quad \bx \in \R^3\,.
\end{align}
It follows from \eqref{eq:V-phi-est} (and that \(0\le 1-\widetilde{\chi}_{\rho}\le
1-\theta_{\rho/2}\)), \eqref{est:with-r}, \eqref{est:beta-fak}, and   
the choice of \(C\) and \(B\) (see
Remark~\ref{rem:constants}), that 
for all \(i=1,\ldots,N\) (recall that \(|\beta|=j+1\))
\begin{align}\label{sec-est-bis-4}\nonumber
   \|\Phi D^{\beta} E(\bp)^{-1}[\eta_{j} V
   (1&-\widetilde{\chi}_{\rho})\varphi_{i}]\|_{p}
   \\&\le \frac{4\sqrt{2}}{\pi}\frac{16}{\pi|\bx_0|}
   \|\varphi_{i}\|_{3}
   \Big(\frac{32}{\epsilon}\Big)^{|\beta|}
   \le\frac{C}{24}\Big(\frac{B}{\epsilon}\Big)^{j+1}\,.
\end{align}
It follows from \eqref{split-V}, \eqref{sec-est-bis-bis}, and
\eqref{sec-est-bis-4}
that
\begin{align}\label{eq:last-term-V-final}
  \|\Phi D^{\beta} E(\bp)^{-1}[\eta_{j} V
  \varphi_{i}]\|_{p}
  \le \frac{C}{12}\Big(\frac{B}{\epsilon}\Big)^{j+1}\,.
\end{align} 

The estimate \eqref{est-V} now follows from \eqref{f2} and the estimates
\eqref{eq:final-first-sum-V}, 
\eqref{eq:second-sum-total-final-V}, and
\eqref{eq:last-term-V-final}. 
\end{pf*}
\begin{pf*}{Proof of \eqref{est-V-bis}}
Note that the constant functions
\(W_{i}(\bx)=\alpha^{-1}+\varepsilon_i\) trivially satisfies 
the conditions on \(V\) (\(=Z\alpha|\cdot|^{-1}\)) needed in the
proof above.
In fact, having assumed \(A\ge \alpha^{-1}+\max_{1\le i\le
  N}|\varepsilon_{i}|\) (See Remark~\ref{rem:constants}),
\eqref{Wanal0-bis} (and therefore \eqref{Wanal}) trivially 
holds for \(W_{i}\). Also, for the term 
\(\Phi D^{\beta} E(\bp)^{-1}[\eta_{j} W_{i} \varphi_{i}]\)
we proceed directly as for the term 
\(\Phi D^{\beta}
E(\bp)^{-1}[\eta_jV(1-\widetilde{\chi}_{\rho})\varphi_i]\) above
(but without any splitting in \(\widetilde{\chi}_{\rho}\) and
\(1-\widetilde{\chi}_{\rho})\), 
using that \(|W_{i}(\bx)|\le A\), \(\bx \in \R^3\).
The proof of \eqref{est-V-bis} therefore follows from the proof of
\eqref{est-V} above, by the choice  of \(C\) and \(B\) (see
Remark~\ref{rem:constants}). 

This finishes the proof of Lemma~\ref{lem:V}.
\end{pf*}
\begin{remark}\label{rem:exp-decay}
In fact, with a simple modification the arguments
above (the local \(L^{p}\)-bound on the two terms in \eqref{split-V})
can be made to work just assuming that, for all \(s>0\), 
\begin{align}\label{eq:L-r-cond-V}
  V\varphi_{i}&\in L^{1}(B_{s}(0))\ , \quad 
  V\varphi_{i}\in L^{3}(\R^{3}\setminus B_{s}(0))\,.
\end{align}
\end{remark}
\section{Proof of  Lemma~\ref{lem:U-phi}}\label{sect:proof-U}
\begin{pf*}{Proof of \eqref{est:U-lemma}}
Similarly to the case of the term with \(V\) in Lemma~\ref{lem:V}, we
here use
the localization functions introduced in 
Appendix~\ref{localization} below. 
With the notation as in the previous section (in particular,
\(\beta=\sigma+e_{\nu}\) with \(|\sigma|=j\)),
Lemma~\ref{Edgardo} 
(with $\ell=j$) implies that
\begin{align}\notag
  \Phi D^{\beta} E(\bp)^{-1}[U_{a,b} \varphi_{i}]
  & = \sum_{k=0}^{j}\Phi D_{\nu}E(\bp)^{-1}D^{\beta_{k}}\chi_{k} 
  D^{\sigma-\beta_{k}} [U_{a,b}\varphi_{i}]  \\ 
  &{\ }+ \sum_{k=0}^{j-1} 
   \Phi D_{\nu}E(\bp)^{-1}D^{\beta_{k}}[\eta_{k},D^{\mu_{k}}] 
  D^{\sigma-\beta_{k+1}} [U_{a,b} \varphi_{i}] \notag \\ \label{h2} 
  &{\ }+  
    \Phi D_{\nu}E(\bp)^{-1}D^{\sigma}[\eta_{j}U_{a,b}\varphi_{i}]\,,
\end{align}
as an identity in \(H^{-|\beta|}(\R^{3})\).
As in  the proof of Lemma~\ref{lem:V}, 
\([\,\cdot\,,\,\cdot\,]\) denotes the commutator, 
\(|\beta_k|=k\), \(|\mu_k|=1\), and \(0\le\eta_k,\chi_k\le1\). (For
the support properties of \(\eta_k,\chi_k\), see the mentioned appendix.)
As in the previous section, we will prove that each term on the right
side of \eqref{h2}  
belong to \(L^{p}(\R^{3})\), and bound their norms.
The claim of the lemma will follow by summing these bounds.\\

\noindent\textit{The first sum in \eqref{h2}.} We first proceed like
for the similar sum in the proof of Lemma~\ref{lem:V} (see
\eqref{f3}, and after). Let $\theta_{k}$ be the characteristic function of
the support of $\chi_{k}$. 
It follows from the induction hypothesis, using that
\(-\Delta U_{a,b} = 4\pi \varphi_{a}\overline{\varphi_{b}}\), and
Theorems~\ref{lemmaSobolev}  and \ref{elliptic},  that the    
\(D^{\sigma-\beta_{k}} [U_{a,b}\varphi_{i}]\)'s belong to \(L^{p}(\omega')\)
for any \(\omega'\subset\subset\omega\). 
As before, the operator \(\Phi D_{\nu}E(\bp)^{-1}D^{\beta_{k}}\chi_{k}\)
is bounded on \(L^{p}(\R^{3})\).
Then, for $k \in \{0,
\dots,j\}$,
\begin{align}\label{h3}\notag 
  \| \Phi D_{\nu}&E(\bp)^{-1} D^{\beta_{k}}\chi_{k}
  D^{\sigma-\beta_{k}} [U_{a,b} \varphi_{i}]\|_{p} \\ 
  & = 
  \|(\Phi E(\bp)^{-1} D_{\nu}D^{\beta_{k}} \chi_{k})
  \theta_{k}  D^{\sigma-\beta_{k}}[
  U_{a,b} \varphi_{i}]\|_{p}\notag \\
  & \leq  
  \|\Phi E(\bp)^{-1} D_{\nu}D^{\beta_{k}} \chi_{k}\|_{\mathcal{B}_{p}}
  \| \theta_{k}  D^{\sigma-\beta_{k}} [
  U_{a,b} \varphi_{i}] \|_{p}\,. 
\end{align}

The first factor on the right side of \eqref{h3} was estimated in the
proof of Lemma~\ref{lem:V} (see \eqref{eq:est-smooth-factor-BIS}):
For all \(k\in\{0,\ldots,j\}\), 
\(\nu\in\{1,2,3\}\),
\begin{align}\label{eq:est-smooth-factor}
  \|\Phi E(\bp)^{-1} D_{\nu}D^{\beta_{k}} \chi_{k}\|_{\mathcal{B}_{p}}
  \le C_{2}\Big(\frac{8}{\epsilon}\Big)^{k}\,,
\end{align} 
 with \(C_{2}\) the constant in \eqref{eq:C-111}.

It remains to estimate the second factor in \eqref{h3}. 
For $k=j$, since $\beta_{j}=\sigma$, we find that, by
\eqref{eq:bound-U-a,b} and the choice of \(C\) and \(B\) (see
Remark~\ref{rem:constants}), 
\begin{equation}\label{h4}
  \| \theta_{j} U_{a,b} \varphi_{i}
  \|_{p} 
  \leq \|U_{a,b}\|_{\infty}
  \|\varphi_{i}\|_{L^{p}(\omega)}
  \le C_{1}\,C
  \le C\Big(\frac{B}{\epsilon}\Big)^{1/2}\,.
\end{equation}
In the last inequality we also used that \(\epsilon\le1\)
(since \(\epsilon(j+1)\le R/2<1\)). 

The estimate for $k \in \{0, \dots, j-1\}$ is more involved. 
We get, by Leibniz's rule, that
 \begin{align}\notag
   \| \theta_{k}  D^{\sigma-\beta_{k}}&[U_{a,b}
   \varphi_{i}] 
   \|_{p} \\
   & \leq \sum_{\mu \leq \sigma-\beta_{k}} {\sigma-\beta_{k} \choose
   \mu}\| \theta_{k} (D^{\mu} U_{a,b})
  (D^{\sigma-\beta_{k}-\mu} 
   \varphi_{i}) \|_{p}
  \label{h5}\,.  
\end{align}
We estimate separately each term on the right side of \eqref{h5}. 

We separate into two cases.

If $\mu=0$ then, using the induction hypothesis (i.e.,
\(\mathcal{P}(p,j-k)\); recall that
\(\supp\,\theta_k\subseteq \omega_{\epsilon(j-k)}\)) and
\eqref{eq:bound-U-a,b},  
\begin{equation}\label{eq:mu=0}
  \| \theta_{k} U_{a,b}D^{\sigma-\beta_{k}} \varphi_{i}
  \|_{p} \leq 
  C_{1}C\Big(\frac{B}{\epsilon}\Big)^{j-k}
  \le \frac{C}{2}\Big(\frac{B}{\epsilon}\Big)^{j-k+1/2}\,.
\end{equation}
In the last inequality we used the choice of \(B\) (see
Remark~\ref{rem:constants}) and that
\(\epsilon\le1\).

If $0<\mu\le\sigma-\beta_{k}$, then 
(since \(\supp\,\chi_k\subseteq\omega_{\epsilon(j-k+1/4)}\))
H\"older's
inequality (with \(1/p=1/(3p)+2/(3p)\)) and
Corollary~\ref{cor:adams} 
give that  
\begin{align}\label{eq:big1}
  \nonumber
  \| \theta_{k} &(D^{\mu} U_{a,b})
  (D^{\sigma-\beta_{k}-\mu}
  \varphi_{i})\|_{p} 
  \\&{}\leq\| \theta_{k} D^{\mu}
  U_{a,b}\|_{3p/2}  \; \| \theta_{k}
  D^{\sigma-\beta_{k}-\mu} 
  \varphi_{i} \|_{3p}
  \nonumber\\ 
  &{}\leq  K_{2}\|D^{\mu} U_{a,b}\|_{L^{3p/2}(\omega_{\epsilon(j-k+1/4)})} 
  \nonumber\\&\quad\times\| D^{\sigma-\beta_{k}-\mu}
  \varphi_{i}\|_{W^{1,p}(\omega_{\epsilon(j-k+1/4)})}^{\theta}
  \|D^{\sigma-\beta_{k}-\mu} 
  \varphi_{i}\|_{L^{p}(\omega_{\epsilon(j-k+1/4)})}^{1-\theta}\,.
  \end{align}
Here, \(K_{2}\) is the constant in Corollary~\ref{cor:adams}, and
\(\theta=2/p<1\). Note that 
\(\omega_{\epsilon(j-k+1/4)}=B_{r}(\bx_{0})\) with
\(r\in[R/2,1]\), since \(\epsilon(j+1)\le R/2\) and
\(R=\min\{1,|\bx_{0}|/4\}\)

We will use Lemma~\ref{lemanR-bis} below
to bound the first factor in \eqref{eq:big1}. The
last two factors we now bound using the induction hypothesis.

If \(\mu\in\N_0^3\) is such that
\(0<\mu\le\sigma-\beta_k\), then the induction hypothesis (in the
form discussed in Remark~\ref{ihy}) 
gives (recall here \eqref{def:Sob-norm} and 
that \(|\sigma|=j, |\beta_k|=k\)) that for the last two factors in
 \eqref{eq:big1} we have
\begin{align}\label{eq:big4}
  \nonumber  
  \| D^{\sigma-\beta_{k}-\mu}
  &\varphi_{i}\|_{L^{p}(\omega_{\epsilon(j-k+1/4)})}^{1-\theta}
  \nonumber \\&
  \le
  \Big[C
  \Big(\frac{j-k-|\mu|}{j-k+1/4}\Big)^{j-k-|\mu|}
  \Big(\frac{B}{\epsilon}\Big)^{j-k-|\mu|}\Big]^{1-\theta}
\end{align}
and (using that \(B>1\) (see Remark~\ref{rem:constants})
and \(\epsilon(j-k+1/4)\le\epsilon(j+1)\le R/2<1\)) 
\begin{align}\label{eq:big3}
  \nonumber 
  \| D^{\sigma-\beta_{k}-\mu}
  &\varphi_{i}\|_{W^{1,p}(\omega_{\epsilon(j-k+1/4)})}^{\theta}
  \le
  \Big[C
  \Big(\frac{j-k-|\mu|}{j-k+1/4}\Big)^{j-k-|\mu|}
  \Big(\frac{B}{\epsilon}\Big)^{j-k-|\mu|}
  \nonumber \\&\qquad\qquad\quad
  +3C\Big(\frac{j-k-|\mu|+1}{j-k+1/4}\Big)^{j-k-|\mu|+1}
  \Big(\frac{B}{\epsilon}\Big)^{j-k-|\mu|+1}\Big]^{\theta}
  \nonumber \\&\quad\le \Big[4C
  \Big(\frac{j-k-|\mu|+1}{j-k+1/4}\Big)^{j-k-|\mu|+1}
  \Big(\frac{B}{\epsilon}\Big)^{j-k-|\mu|+1}\Big]^{\theta}\,.
\end{align}
It follows from 
\eqref{eq:big4} and \eqref{eq:big3}
that for all  \(\mu\in\N_0^3\) with 
\(0<\mu\le\sigma-\beta_k\),
\begin{align}\label{est:final-last-two-factors}\nonumber
  \|D^{\sigma-\beta_{k}-\mu}\varphi_{i}&\|_{W^{1,p}
  (\omega_{\epsilon(j-k+1/4)})}^{\theta}  
  \|D^{\sigma-\beta_{k}-\mu}
  \varphi_{i}\|_{L^{p}(\omega_{\epsilon(j-k+1/4)})}^{1-\theta}
  \\&\le C4^{\theta}  \Big(\frac{B}{\epsilon}\Big)^{j-k-|\mu|+\theta}
  \Big(\frac{j-k-|\mu|+1}{j-k+1/4}\Big)^{j-k-|\mu|+\theta}\,.
\end{align}

From \eqref{eq:big1}, Lemma~\ref{lemanR-bis}, and
\eqref{est:final-last-two-factors} (using \eqref{eq:multiNom} in 
Appendix~\ref{Notation} below,  summing over \(m=|\mu|\)), it follows
that  
\begin{align}\label{eq:1st}\nonumber
  \sum_{0<\mu\le\sigma-\beta_k}
  &\binom{\sigma-\beta_k}{\mu}
  \| \theta_{k} (D^{\mu} U_{a,b})
  (D^{\sigma-\beta_{k}-\mu}
  \varphi_{i})\|_{p} 
  \le C^{3} C_{3} K_{2}
  \Big(\frac{B}{\epsilon}\Big)^{j-k+\theta}\times
  \\\nonumber
  &\times\sum_{m=1}^{j-k}  4^{\theta}\binom{j-k}{m}
  \frac{(j-k-m+1)^{j-k-m+\theta}
  (m+1/4)^{m}}{(j-k+1/4)^{j-k+\theta}}\times
  \\&\qquad\quad\quad
  \times\Big[\Big(\frac{1}{\sqrt{B}}\Big)^{m} 
  +\sqrt{m}
  \Big(\frac{B(m+1/4)}{\epsilon(j-k+1/4)}\Big)^{2\theta-2}  
  \Big]\,.
\end{align}
Here, \(C_{3}\) is the constant from \eqref{eq:const-lem7}. Recall also
that \(\theta=2/p\). 

We prove that for \(m\in\{1,\ldots,j-k\}\), 
\begin{align}\label{eq:new-est-SF}
  4^{\theta}\binom{j-k}{m}
  \frac{(j-k-m+1)^{j-k-m+\theta}
  (m+1/4)^{m}}{(j-k+1/4)^{j-k+\theta}}
  \le 10\epsilon^{-1/2+\theta}\frac{1}{\sqrt{m}}\,.
\end{align}
Note first that, since \(\epsilon(j-k+1/4)\le \epsilon(j+1)\le1\),  
\begin{align}\label{eq:small-one-bis}
  (j-k+1/4)^{1/2-\theta}
  \le \epsilon^{-1/2+\theta}\,.
\end{align}
This shows that the
inequality in \eqref{eq:new-est-SF} is true for \(m=j-k>0\), since
\(\theta<1\). For \(m<j-k\), we use \eqref{eq:est-binom-intro}
in Appendix~\ref{Notation} below, and \eqref{eq:small-one-bis},
to get that (since \((1+1/n)^n\le
{\rm e}\)) 
\begin{align}\label{eq:3rd}\nonumber
  \binom{j-k}{m}& \frac{(j-k-m+1)^{j-k-m+\theta}
   (m+1/4)^{m}}{(j
    -k+1/4)^{j-k+\theta}}
  \\&\quad\quad
  \le\frac{{\rm e}^{25/12}}{\sqrt{2\pi}}
  \frac{(j-k-m+1)^{\theta}}{(j-k-m)^{1/2}}
  \epsilon^{-1/2+\theta}\frac{1}{\sqrt{m}} \,.
\end{align}
Since \(\theta<1/2\) and \(m\le j-k-1\), 
we have that
\begin{align}\label{eq:small-one}
  \frac{(j-k-m+1)^{\theta}}{(j-k-m)^{1/2}}\le 2^{\theta}\le\sqrt2\,.
\end{align}
The estimate \eqref{eq:new-est-SF} for \(m\in\{1,\ldots,j-k-1\}\) now
follows from \eqref{eq:3rd}--\eqref{eq:small-one} (since \(4^{\theta}{\rm
  e}^{25/12}/\sqrt{\pi}\le10\)).  

Inserting \eqref{eq:new-est-SF} in \eqref{eq:1st} (and using again
\(\epsilon(j-k+1/4)\le1\) and \(2\theta-2<0\)) we find that 
\begin{align}\label{eq:1st-BIS}\nonumber
   \sum_{0<\mu\le\sigma-\beta_k}
   &\binom{\sigma-\beta_k}{\mu}
   \| \theta_{k} (D^{\mu} U_{a,b})
   (D^{\sigma-\beta_{k}-\mu}
   \varphi_{i})\|_{p} 
   \\\nonumber
   &\le 10C^{3} C_{3} K_{2}
   \Big(\frac{B}{\epsilon}\Big)^{j-k+\theta}
   \epsilon^{-1/2+\theta}
   \sum_{m=1}^{j-k}\Big[\Big(\frac{1}{\sqrt{B}}\Big)^{m}
   +\frac{1}{B^{2-2\theta}}\frac{1}{m^{2-2\theta}}\Big]
   \\&\le
   10C^{3} C_{3} K_{2}
   \Big(\frac{B}{\epsilon}\Big)^{j-k+1/2}
   \frac{1}{\sqrt{B}}\,(2+6)\,,
\end{align}
where we used that \(\theta\le 2/5\), \(B\ge 4\) (see
Remark~\ref{rem:constants}), and
\(\sum_{m=1}^{\infty}m^{-6/5}\le1+\int_{1}^{\infty}x^{-6/5}\,dx=6\) 
 to estimate 
\begin{align}\label{eq:helping}
  \sum_{m=1}^{\infty}\Big(\frac{1}{\sqrt{B}}\Big)^{m}\le
  \frac{2}{\sqrt{B}}\,,
  \quad
  \frac{1}{B^{2-2\theta}}\sum_{m=1}^{\infty}
  \frac{1}{m^{2-2\theta}}
  \le \frac{6}{\sqrt{B}}\,.
\end{align}
This is the very essential reason for needing \(p\ge 5\).

By the choice of \(B\) (see
Remark~\ref{rem:constants})
it follows that 
\begin{align}\label{eq:2nd}
  \sum_{0<\mu\le\sigma-\beta_k}
  \binom{\sigma-\beta_k}{\mu}
  \| \theta_{k} (D^{\mu} U_{a,b})
  (D^{\sigma-\beta_{k}-\mu}
  \varphi_{i})\|_{p} 
  \le
  \frac{C}{2}\Big(\frac{B}{\epsilon}\Big)^{j-k+1/2}\,.
\end{align}

From \eqref{h5}, \eqref{eq:mu=0}, and  
\eqref{eq:2nd} it follows that for all \(k\in\{0,\ldots,j-1\}\), 
\begin{align}\label{eq:final-theta-k}
  \| \theta_{k}  D^{\sigma-\beta_{k}}[U_{a,b}
  \varphi_{i}]\|_{p} 
  \le
  C\Big(\frac{B}{\epsilon}\Big)^{j-k+1/2}\,.
\end{align}

Using \eqref{h3},
\eqref{eq:est-smooth-factor}, \eqref{h4}, and \eqref{eq:final-theta-k}
it follows for the first sum in \eqref{h2}
that
\begin{align}\label{eq:first-sum-total}\nonumber
  &\Big\|\sum_{k=0}^{j} 
  \Phi D_{\nu}E(\bp)^{-1}D^{\beta_{k}}\chi_{k} 
  D^{\sigma-\beta_{k}} [U_{a,b}\varphi_{i}]\Big\|_{p}
  \\&\le C_{2} \sum_{k=0}^{j}
  8^{k}\epsilon^{-k}
  \|\theta_kD^{\sigma-\beta_k}[U_{a,b}
  \varphi_{i}]\|_{p}
  \le C_{2}C
  \Big(\frac{B}{\epsilon}\Big)^{j+1/2}
  \sum_{k=0}^{j}\Big(\frac{8}{B}\Big)^k\,.
\end{align}
Since  \(B>16\) 
(see Remark~\ref{rem:constants}) the last sum is less
than \(2\) and so for the first term in \eqref{h2} we finally get, by
the choice of \(B\) (see Remark~\ref{rem:constants}) that
\begin{align}\label{eq:final-first-term}\nonumber
   \Big\|\sum_{k=0}^{j}   
   \Phi D_{\nu}E(\bp)^{-1}D^{\beta_{k}}\chi_{k} 
  &D^{\sigma-\beta_{k}} [U_{a,b}\varphi_{i}]\Big\|_{p}
   \\& \le 2C_{2}C
   \Big(\frac{B}{\epsilon}\Big)^{j+1/2}
  \le \frac{CZ}{12N}\Big(\frac{B}{\epsilon}\Big)^{j+1}\,.
\end{align}

\ %

\noindent\textit{The second sum in \eqref{h2}.}
By the same arguments as for the second sum in \eqref{f2} (see
after \eqref{eq:final-first-sum-V}), 
it follows that the second sum in \eqref{h2} can be
estimated as the first one, up to {\it one} extra
factor of \(C_*/\epsilon\) (with
\(C_{*}>0\) the constant in \eqref{eq:est-der-loc} in
Appendix~\ref{localization} below) {\it and} up to replacing \(j\) by
\(j-1\) in 
the estimate \eqref{eq:final-first-term}. Hence, by the choice of
\(B\) (see Remark~\ref{rem:constants}) 
\begin{align}\label{eq:second-sum-total-final}\nonumber
  \Big\|\sum_{k=0}^{j-1} 
  \Phi D_{\nu}E(\bp)^{-1}D^{\beta_{k}}
  &[\eta_{k},D^{\mu_{k}}] D^{\sigma-\beta_{k+1}}[
  U_{a,b} \varphi_{i}]\Big\|_{p}
  \\&
  \le  \frac{C_{*}}{\epsilon}\frac{CZ}{12N}
  \Big(\frac{B}{\epsilon}\Big)^{j}
  \le \frac{CZ}{12N}\Big(\frac{B}{\epsilon}\Big)^{j+1}\,.
\end{align}

\ %

\noindent\textit{The last term in \eqref{h2}.} 
Since \(\sigma+e_{\nu}=\beta\), the last term in \eqref{h2} equals
\begin{align*}
  \Phi D^{\beta} E(\bp)^{-1}[\eta_{j}U_{a,b}
  \varphi_{i}]\,.
\end{align*}
We proceed exactly as for the term \(\Phi D^{\beta}
E(\bp)^{-1}[\eta_jV(1-\widetilde{\chi}_{\rho})\varphi_i]\) in 
\eqref{split-V} (but without any splitting in \(\widetilde{\chi}_{\rho}\) and
\(1-\widetilde{\chi}_{\rho})\), except that the estimate in
\eqref{eq:V-phi-est} is replaced by \(\|U_{a,b}\|_{\infty}\le
C_{1}\) (see \eqref{eq:bound-U-a,b}).
It follows, from the choice of \(B\) and \(C\) (see
Remark~\ref{rem:constants}) that (recall that \(|\beta|=j+1\))
 \begin{align}\label{sec-est-bis-3}   \nonumber
   \|\Phi &D^{\beta} E(\bp)^{-1}[\eta_{j}
   U_{a,b}\varphi_{i}]\|_{p}
   \le \|\Phi D^{\beta} E(\bp)^{-1}\eta_{j}\|_{\mathcal{B}_{3,p}}
   \| U_{a,b}\varphi_{i}\|_{3}
   \\&\le
   \frac{4\sqrt{2}}{\pi}C_{1}
   \|\varphi_i\|_{3}\Big(\frac{32}{\epsilon}\Big)^{|\beta|}
   \le\frac{CZ}{12N}\Big(\frac{B}{\epsilon}\Big)^{j+1}\,.
\end{align}
The estimate \eqref{est:U-lemma} now follows from \eqref{h2} and the estimates
\eqref{eq:final-first-term},
\eqref{eq:second-sum-total-final}, and
\eqref{sec-est-bis-3}.

This finishes the proof of Lemma~\ref{lem:U-phi}.
\end{pf*}
It remains to prove
Lemma~\ref{lemanR-bis} below
(\(L^{3p/2}\)-bound on derivatives of the 
Newton potential \(U_{a,b}\) of
products of orbitals, \(\varphi_a\varphi_b\)).

In the next lemma we first give an $L^{3p/2}$-estimate on the
derivatives of the product of the orbitals \(\varphi_{i}\), needed for
the proof of the bound
in Lemma~\ref{lemanR-bis} below.
\begin{lemma}\label{prL3}
Assume \eqref{eq:ihy} (the induction hypothesis) holds. Then, for all
$a,b \in \{1, \dots, N\}$, all $\beta \in \N_{0}^3$ with $|\beta|\leq
j-1$, and all $\epsilon>0$ with \(\epsilon(|\beta|+1)\le R/2\),
\begin{align}
  \label{est:R-term}
  \| D^{\beta}(\varphi_{a}\overline{\varphi_{b}})
  \|_{L^{3p/2}(\omega_{\epsilon(|\beta|+1)})} &\leq 10K_{2}^2 C^2
   (1+\sqrt{|\beta|})\Big(\frac{B}{\epsilon}\Big)^{|\beta|+2\theta}\,,
\end{align}
with \(K_{2}\) from Corollary~\ref{cor:adams}, \(C\) from
Remark~\ref{rem:constants}, and \(\theta=\theta(p)=2/p\). 
\end{lemma}
\begin{proof} 
By Leibniz's rule and the Cauchy-Schwarz inequality we get that
\begin{align*}
  \| D^{\beta}(\varphi_{a}&\overline{\varphi_{b}})
  \|_{L^{3p/2}(\omega_{\epsilon(|\beta|+1)})} 
  \\&\leq  
  \sum_{\mu \leq \beta}{\beta \choose \mu} \| D^{\mu}
  \varphi_a \|_{L^{3p}(\omega_{\epsilon(|\beta|+1)})} 
  \| D^{\beta-\mu} \varphi_b \|_{L^{3p}(\omega_{\epsilon(|\beta|+1)})}. 
\end{align*}
We use Corollary~\ref{cor:adams} 
(with \(\omega_{\epsilon(|\beta|+1)}=
B_{r}(\bx_{0})\),
\(r=R-\epsilon(|\beta|+1)\); 
note that \(r\in[R/2,1]\), since
\(\epsilon(|\beta|+1)\le R/2\) and  
\(R=\min\{1,|\bx_{0}|/4\}\)). This gives that, with \(K_{2}\) from
Corollary~\ref{cor:adams}, and \(\theta=2/p\), 
\begin{align}\label{eq:1}\nonumber
  \| D^{\beta}(\varphi_{a} &\overline{\varphi_{b}})
  \|_{L^{3p/2}(\omega_{\varepsilon(|\beta|+1)})} 
  \\ 
  &\leq  K_{2}^2 \sum_{\mu \leq \beta} {\beta \choose \mu} 
  \| D^{\mu}
  \varphi_a\|_{W^{1,p}(\omega_{\epsilon(|\beta|+1)})}^{\theta} 
  \|D^{\mu}\varphi_{a}\|_{L^{p}(\omega_{\epsilon(|\beta|+1)})}^{1-\theta}
  \\&\qquad\qquad\quad\ \,\times
   \| D^{\beta-\mu}
   \varphi_b\|_{W^{1,p}(\omega_{\epsilon(|\beta|+1)})}^{\theta}  
  \|D^{\beta-\mu}\varphi_{b}\|_{L^{p}(\omega_{\epsilon
  (|\beta|+1)})}^{1-\theta}\,.
  \nonumber
\end{align}
We now use the induction hypothesis (in the form
discussed in Remark~\ref{ihy}) on each of the four factors in the sum
on the right side of \eqref{eq:1}. Note that, by assumption,
\(\epsilon(|\beta|+1)\le \epsilon j\le R/2\) and \(|\mu|<|\mu|+1\le
|\beta|+1\le j\)
(similarly, \(|\beta-\mu|<|\beta-\mu|+1\le j\)). 
Recalling 
\eqref{def:Sob-norm}, we therefore get that, for all
\(\mu\in\N_{0}^{3}\) such that \(\mu\le \beta\), 
\begin{align*} 
 \| &D^{\mu}\varphi_a\|_{W^{1,p}(\omega_{\epsilon(|\beta|+1)})}^{\theta} 
 \|D^{\mu}\varphi_{a}\|_{L^{p}(\omega_{\epsilon(|\beta|+1)})}^{1-\theta}
 \\&\le \Big[C\Big(\frac{|\mu|}{|\beta|
         +1}\Big)^{|\mu|}\Big(\frac{B}{\epsilon}\Big)^{|\mu|}
 \Big]^{1-\theta}
  \\&\ \ \times
    \Big[C\Big(\frac{|\mu|}{|\beta|
         +1}\Big)^{|\mu|}\Big(\frac{B}{\epsilon}\Big)^{|\mu|}
     +3C\Big(\frac{|\mu|+1}{|\beta|
         +1}\Big)^{|\mu|+1}\Big(\frac{B}{\epsilon}\Big)^{|\mu|+1}
  \Big]^{\theta}
  \\&\le 4^{\theta}C\Big(\frac{B}{\epsilon}\Big)^{|\mu|+\theta}
  \frac{(|\mu|+1)^{\theta(|\mu|+1)}
         |\mu|^{|\mu|(1-\theta)}}{(|\beta|+1)^{|\mu|+\theta}}\,,
\end{align*}
since (recall that \(\epsilon(|\beta|+1)\le R/2<1\) and
\(B>1\)) 
\begin{align*}
  \frac{|\mu|^{|\mu|}}{(|\mu|+1)^{|\mu|+1}}
  \,\epsilon(|\beta|+1)B^{-1}
  \le 1\,.
\end{align*}
Proceeding similarly for the other two factors in \eqref{eq:1}, we get
(using \eqref{eq:multiNom} 
in Appendix~\ref{Notation}
and summing over \(m=|\mu|\)) that
\begin{align}\label{eq:first-product}\nonumber 
  &\sum_{\mu\le\beta} {\beta \choose \mu}
   \| D^{\mu} \varphi_a \|_{L^{3p}(\omega_{\epsilon(|\beta|+1)})} 
  \| D^{\beta-\mu} \varphi_b
  \|_{L^{3p}(\omega_{\epsilon(|\beta|+1)})}
  \\&\le
   16^{\theta}(CK_{2})^2\Big(\frac{B}{\epsilon}\Big)^{|\beta|+2\theta}\,
  \times
  \\&
  \sum_{m=0}^{|\beta|} \binom{|\beta|}{m}
  \frac{\big[(m+1)^{m+1}(|\beta|-m+1)^{|\beta|-m+1}\big]^{\theta}
     \big[m^m(|\beta|-m)^{|\beta|-m}\big]^{1-\theta}}{
    (|\beta|+1)^{|\beta|+2\theta}}\,.\nonumber
\end{align}
We simplify the sum in $m$. Note that for \(m=0\) and \(m=|\beta|\),
the summand is bounded by \(1\).
Therefore, for \(|\beta|\le 1\) the estimate \eqref{est:R-term} follows from
\eqref{eq:first-product}, since \( 2\cdot 16^{\theta} \leq 7\).
It remains to consider \(|\beta|\ge2\).
For $m \geq 1$, \(m<|\beta|\), 
we can use \eqref{eq:est-binom-intro} in  Appendix~\ref{Notation}
to get (since \((1+1/n)^n\le{\rm e}\)) that
\begin{align*}
  \sum_{0<\mu<\beta} {\beta \choose \mu}
   \| D^{\mu} &\varphi_a \|_{L^{3p}(\omega_{\epsilon(|\beta|+1)})} 
  \| D^{\beta-\mu} \varphi_b
  \|_{L^{3p}(\omega_{\epsilon(|\beta|+1)})}
  \\&\le \frac{{\rm e}^{1/12}}{\sqrt{2\pi}}
  (CK_{2})^2(16{\rm e}^2)^{\theta}\Big(\frac{B}{\epsilon}\Big)^{
  |\beta|+2\theta}
  \frac{|\beta|^{|\beta|+1/2}}{(|\beta|+1)^{|\beta|+2\theta}}
  \\&\qquad\qquad\qquad\quad\times
  \sum_{m=1}^{|\beta|-1}\frac{\big[(m+1)(|\beta|-m+1)\big]^{\theta}}{
  \sqrt{m}\sqrt{|\beta|-m}}\,.
\end{align*}
Since the function
\begin{align*}
  f(x)=(x+1)(|\beta|-x+1)\,, \quad x\in[1,|\beta|-1]\,,
\end{align*}
has its maximum (which is \((|\beta|/2+1)^2\)) at $x=|\beta|/2$, and
since   
\begin{equation*}
  \sum_{m=1}^{|\beta|-1}\frac{1}{\sqrt{m}\sqrt{|\beta|-m}}
  \le\int_{0}^{|\beta|} \frac{1}{\sqrt{x} \sqrt{|\beta|-x}}\,dx
  =\pi\,,
\end{equation*}
we get that
\begin{align}\label{eq:second-product}\nonumber
  \sum_{0<\mu<\beta} {\beta \choose \mu}
   &\| D^{\mu} \varphi_a \|_{L^{3p}(\omega_{\epsilon(|\beta|+1)})} 
  \| D^{\beta-\mu} \varphi_b
  \|_{L^{3p}(\omega_{\epsilon(|\beta|+1)})}
  \\&\le {\rm e}^{1/12} (16{\rm
    e}^2)^{\theta}\sqrt{\frac{\pi}{2}}
  (CK_{2})^2  \sqrt{|\beta|}
  \Big(\frac{B}{\epsilon}\Big)^{|\beta|+2\theta}\,.
\end{align}
The estimate \eqref{est:R-term} now follows from
\eqref{eq:1},
\eqref{eq:first-product}, and \eqref{eq:second-product}, since (as
\(p\ge5\)), 
\begin{align*}
  {\rm e}^{1/12}(16{\rm e}^2)^{\theta}\sqrt{\frac{\pi}{2}}
  \le 10\,,\qquad
  2\cdot 16^{\theta}\le 7\,.
\end{align*}
This finishes the proof of Lemma~\ref{prL3}.
\end{proof}
The next two lemmas, used in the proof above of Lemma~\ref{lem:U-phi},
 control the \(L^{3p/2}\)-norm of derivatives of
\(U_{a,b}\).
\begin{lemma}\label{lemanR} 
Define  $U_{a,b}$ by \eqref{eq:def-U-a,b}.
Then for all $a,b \in \{1, \dots, N\}$,  
 and all $\mu \in \mathbb{N}_{0}^3$ with \(|\mu|\le2\), 
\begin{align}\label{eq:est-pot-mu=1}
  \| D^{\mu} U_{a,b}&\|_{L^{3p/2}(\omega)} 
  \le 4\pi K_3(C^2+2C_{1}/R^2)\,, 
\end{align}
with \(K_{3}\) from Corollary~\ref{rem:elliptic}, \(C\) from
Remark~\ref{rem:constants}, 
\(C_{1}\) from
\eqref{eq:bound-U-a,b}, and \(R=\min\{1,|\bx_{0}|/4\}\).
\end{lemma}
\begin{pf}
Recall that \(\omega=B_{R}(\bx_{0})\),
\(R=\min\{1,|\bx_{0}|/4\}\). 
Using \eqref{def:Sob-norm}, and 
Corollary~\ref{rem:elliptic}, we get
that, for all \(\mu\in\N_{0}^{3}\) with \(|\mu|\le2\), 
\begin{align}\label{eq:one-der-U-a,b's}
  \|D^{\mu} U_{a,b}\|_{L^{3p/2}(\omega)}  
  &\le \|U_{a,b}\|_{W^{2,3p/2}(B_{R}(\bx_{0}))}
  \\&\le K_{3}\big\{ \|\Delta  U_{a,b}\|_{L^{3p/2}(B_{2R}(\bx_{0}))}
  +\frac{1}{R^2}\| U_{a,b}\|_{L^{3p/2}(B_{2R}(\bx_{0}))}\big\}\,.
  \nonumber
\end{align}
By the definition of \(U_{a,b}\) (see \eqref{eq:def-U-a,b}) 
we have  
\begin{align}\label{eq:Poisson} 
  {}-\Delta U_{a,b}(\bx) =
  4\pi\,\varphi_{a}(\bx)\overline{\varphi_{b}}(\bx) \ \text{ for } \ \bx \in
  \R^{3}\,, 
\end{align}
and \(\|U_{a,b}\|_{\infty}\le C_{1}\) (see
\eqref{eq:bound-U-a,b}). Hence, from 
\eqref{eq:one-der-U-a,b's}, H{\"o}lder's inequality, and the choice of
\(C\) (see Remark~\ref{rem:constants}; recall also that \(p\ge 5\))
\begin{align*}
   \| D^{\mu} U_{a,b}\|_{L^{3p/2}(\omega)}
   &\le 4\pi K_{3}\big\{
  \|\varphi_{a}\|_{L^{3p}(B_{2R}(\bx_{0}))}\|\varphi_{b}\|_{L^{3p}(B_{2R}(\bx_{0}))}
  \\&\qquad\qquad\qquad +\frac{1}{R^2}\|U_{a,b}\|_{\infty}|B_{2R}(\bx_{0})|^{2/3p}\big\}
  \\&
  \le 4\pi K_{3}(C^2+2C_{1}/R^2)\,.
\end{align*}

This finishes the proof of the lemma.
\end{pf}
\begin{lemma}\label{lemanR-bis}
Assume \eqref{eq:ihy} (the induction hypothesis) holds, 
and define  $U_{a,b}$ by \eqref{eq:def-U-a,b}.

Then 
  for all $a,b \in \{1, \dots, N\}$, all $k \in \{0, \dots,j-1\}$,
  all $\mu \in \mathbb{N}_{0}^3$ with $|\mu| \leq j-k$,  and all
\(\epsilon>0\) with \(\epsilon(j+1)\le R/2\),
\begin{align}\label{eq:der-U(p,q)-bd}
  \nonumber
  \| D^{\mu} &U_{a,b}\|_{L^{3p/2}(\omega_{\epsilon(j-k+1/4)})}
   \leq  C_{3}C^2  \Big(\frac{\sqrt{B}}{\epsilon}\Big)^{|\mu|}
   \Big(\frac{|\mu|+1/4}{j-k+1/4}\Big)^{|\mu|}
 \\&\ +
   C_{3} 
  C^2\sqrt{|\mu|}\Big(\frac{B}{\epsilon}\Big)^{|\mu|+2\theta-2}
  \Big(\frac{|\mu|+1/4}{j-k+1/4}\Big)^{|\mu|+2\theta-2}\,,
\end{align}
with \(\theta=\theta(p)=2/p\), \(C\) and \(B\) from
Remark~\ref{rem:constants}, 
and \(C_{3}\) the constant in \eqref{eq:const-lem7}.
\end{lemma}
\begin{pf} 
If \(m:=|\mu|\le 2\), \eqref{eq:der-U(p,q)-bd} follows from
Lemma~\ref{lemanR} and the definition of \(C_{3}\) in
\eqref{eq:const-lem7}, since \(\epsilon(j-k+1/4)\le \epsilon(j+1)\le
R/2<1\), and \(C, B>1\) (see Remark~\ref{rem:constants}).

If $m:=|\mu|\geq 3$ then we write $\mu =
\mu_{m-2}+e_{\nu_1}+e_{\nu_2}$ with 
\(\nu_i\in\{1,2,3\}, i=1,2\), \(|\mu_{m-2}|=m-2\). 
Then by the definition of the $W^{2,3p/2}$-norm (recall
\eqref{def:Sob-norm}) we find that
\begin{align}\label{eq:first-iteration}\nonumber
  \| D^{\mu} U_{a,b}\|_{L^{3p/2}(\omega_{\epsilon(j-k+1/4)})} 
  &\leq 
  \| D^{\mu_{m-2}} U_{a,b} \|_{W^{2,3p/2}(\omega_{\epsilon(j-k+1/4)})} 
  \\&= \|D^{\mu_{m-2}} U_{a,b}
  \|_{W^{2,3p/2}(\omega_{\tilde{\epsilon}_{1}(m-1+1/4)})} ,  
\end{align}
with \(\tilde{\epsilon}_{1}\) such that
\begin{equation}\label{eq:eps-1}
  \tilde{\epsilon}_{1}(m-1+1/4)= \epsilon(j-k+1/4)\,.
\end{equation}
To estimate the norm in \eqref{eq:first-iteration} we will again use that 
$U_{a,b}$ satisfies \eqref{eq:Poisson}.
Applying $D^{\mu_{m-2}}$ to \eqref{eq:Poisson}
and using  the elliptic
{\it a priori} estimate in
Corollary~\ref{rem:elliptic}
(with \(r=r_1=R-\tilde{\epsilon}_{1}(m-1+1/4)\) and
\(\delta=\delta_1=\tilde{\epsilon}_{1}/4\); recall that
\(\omega_{\rho}=B_{R-\rho}(\bx_{0})\)) 
we get that
\begin{align}\label{eq:first-iterate}\nonumber
  \|D^{\mu} U_{a,b}\|_{L^{3p/2}(\omega_{\epsilon(j-k+1/4)})} 
  &\leq 4\pi K_{3} 
  \|D^{\mu_{m-2}}(\varphi_{a}\overline{\varphi_{b}})
  \|_{L^{3p/2}(\omega_{\tilde{\epsilon}_{1}(m-1)})} 
  \\&\qquad +
  \frac{16K_{3}}{\tilde{\epsilon}_{1}^{2}}  
  \|
  D^{\mu_{m-2}}U_{a,b}
  \|_{L^{3p/2}(\omega_{\tilde{\epsilon}_{1}(m-1)})}\,,
\end{align}
with \(K_{3}=K_{3}(p)\) the constant in \eqref{eq:elliptic-est-bis}. 
Notice that for this estimate we needed to enlarge the domain, taking
the ball with a radius $\tilde{\epsilon}_{1}/4 $ larger.  

We now iterate the procedure (on the second term on the right side of
\eqref{eq:first-iterate}), with 
\(\tilde{\epsilon}_{i}\)  
(\(i=2,\ldots, \lfloor \frac{m}{2} \rfloor\))
 such that
\begin{align}\label{eq:eps-i}
  \tilde{\epsilon}_{i}(m-2i+1+1/4)=\tilde{\epsilon}_{i-1}(m-2(i-1)+1)\,,   
\end{align}
and with \(r=r_{i}=R-\tilde{\epsilon}_{i}(m-2i+1+1/4)\) and 
\(\delta=\delta_{i}=\tilde{\epsilon}_{i}/4\).
Note that \eqref{eq:eps-1} and \eqref{eq:eps-i} imply that,
for \(i=2,\ldots, \lfloor \frac{m}{2} \rfloor\),
\begin{align}\label{eq:est-eps-i}
  &\tilde{\epsilon}_{i}\ge   \tilde{\epsilon}_{i-1}
  \ge\ldots\ge   \tilde{\epsilon}_{1}
  =\epsilon\,\frac{j-k+1/4}{m-1+1/4}\,,
\end{align}
and
\begin{align}\label{eq:formula-new-eps's}\nonumber
  \tilde{\epsilon}_{i}(m-2i+1)
  &\le\tilde{\epsilon}_{i-1}(m-2(i-1)+1)
  \\&\le\ldots\le
  \tilde{\epsilon}_{1}(m-1)\le\epsilon(j-k+1/4)\,. 
\end{align}
We get that 
(with \(\prod_{\ell=1}^{0}\equiv1\) and \(|\mu_{m-2i}|=m-2i\)), 
\begin{align}\label{eq:ugly}\nonumber
  \| D^{\mu} &U_{a,b}\|_{L^{3p/2}(\omega_{\epsilon(j-k+1/4)})} 
   \\&\leq  4\pi K_{3}
  \sum_{i=1}^{\lfloor \frac{m}{2} \rfloor} 
  \Big[
  \| D^{\mu_{m-2i}} (\varphi_{a}\overline{\varphi_{b}})
  \|_{L^{3p/2}(\omega_{\tilde{\epsilon}_{i}(m-2i+1)})}
  \prod_{\ell=1}^{i-1}
  \Big(\frac{16K_{3}}{\tilde{\epsilon}_{\ell}^2}\Big)  
  \Big]
  \nonumber
  \\
  &\quad
  +\Big[\,\prod_{\ell=1}^{\lfloor \frac{m}{2} \rfloor} 
   \frac{16K_{3}}{\tilde{\epsilon}_{\ell}^2}\Big]
  \| D^{\mu_{m-2\lfloor \frac{m}{2} \rfloor}}\,
   U_{a,b}
   \|_{L^{3p/2}(\omega_{\tilde{\epsilon}_{\lfloor\frac{m}{2}\rfloor}
   (m- 2\lfloor\frac{m}{2}\rfloor
   + 1)})}\,.  
\end{align}
Using \eqref{eq:est-eps-i}, and Lemma~\ref{prL3} for 
each \(i=1,\ldots, \lfloor \frac{m}{2} \rfloor\) fixed
(note that \(\tilde{\epsilon}_{i}(m-2i+1)\le R/2\) by
\eqref{eq:formula-new-eps's} since \(\epsilon(j+1)\le R/2\)) 
we get that
\begin{align}\label{eq:use-of-prL3}\nonumber
   \| &D^{\mu_{m-2i}} (\varphi_{a}\overline{\varphi_{b}})
  \|_{L^{3p/2}(\omega_{\tilde{\epsilon}_{i}(m-2i+1)})}
  \prod_{\ell=1}^{i-1}
  \Big(\frac{16K_{3}}{\tilde{\epsilon}_{\ell}^2}\Big)  
   \\&\le
  20K_{2}^2C^2\sqrt{m}\Big(\frac{B}{\epsilon}\Big)^{m+2\theta-2}
  \Big(\frac{m-1+1/4}{j-k+1/4}\Big)^{m+2\theta-2}
  \Big(\frac{16K_{3}}{B^2}\Big)^{i-1}\,,
\end{align}
with \(K_{2}\) from Corollary~\ref{cor:adams}, and
\(\theta=\theta(p)=2/p\). Here we also used that \(1+\sqrt{m-2i}\le
2\sqrt{m}\). 
Note that 
\(\sum_{i=1}^{\lfloor \frac{m}{2} \rfloor}(16K_{3}/B^2)^{i-1}<2\)
since  \(B^2>32K_{3}\) (see Remark~\ref{rem:constants}).
It follows that
\begin{align}\label{eq:est-firsty}
   \nonumber
   4\pi K_{3}
  &\sum_{i=1}^{\lfloor \frac{m}{2} \rfloor} 
  \Big[
  \| D^{\mu_{m-2i}} (\varphi_{a}\overline{\varphi_{b}})
  \|_{L^{3p/2}(\omega_{\tilde{\epsilon}_{i}(m-2i+1)})}
  \prod_{\ell=1}^{i-1}
  \Big(\frac{16K_{3}}{\tilde{\epsilon}_{\ell}^2}\Big)  
  \Big]
  \\&\le
  160\pi K_{2}^2K_{3}C^2\sqrt{m}\Big(\frac{B}{\epsilon}\Big)^{m+2\theta-2}
  \Big(\frac{m+1/4}{j-k+1/4}\Big)^{m+2\theta-2}\,.
\end{align}

We now estimate the last term in \eqref{eq:ugly}. Let
\(\delta=m-2\lfloor \frac{m}{2} \rfloor \in\{0,1\}\) (depending on
whether \(m\) is even or odd). Then, using  \eqref{eq:est-eps-i} and
Lemma~\ref{lemanR}, we get that 
\begin{align}\label{eq:sec-term}
    \nonumber
    \Big[\,\prod_{\ell=1}^{\lfloor \frac{m}{2} \rfloor} 
    &\frac{16K_{3}}{\tilde{\epsilon}_{\ell}^2}\Big]
   \| D^{\mu_{m-2\lfloor \frac{m}{2} \rfloor}}\,
    U_{a,b}
    \|_{L^{3p/2}(\omega_{\tilde{\epsilon}_{\lfloor\frac{m}{2}\rfloor}
    (m- 2\lfloor\frac{m}{2}\rfloor
    + 1)})} 
    \\\nonumber
    &\le 4\pi K_{3}(C^2+2C_{1}/R^2)
    \Big(\frac{\sqrt{16K_{3}}}{\epsilon}\Big)^{m}
    \Big(\frac{m-1+1/4}{j-k+1/4}\Big)^{m}\\\nonumber
    &\qquad\qquad\qquad\qquad\qquad\qquad\qquad\times
    \Big( \frac{\epsilon(j-k+1/4)}{m-1+1/4}\Big)^{\delta}
    \\&\le 4\pi K_{3}(1+2C_{1}/R^2)C^2
   \Big(\frac{\sqrt{B}}{\epsilon}\Big)^{m}
   \Big(\frac{m+1/4}{j-k+1/4}\Big)^{m}\,.
\end{align}
Here we also used that \(m\ge 3\) and \(K_{3}\ge1\) (See Corollary~\ref{rem:elliptic}), 
that \(C>1\) and 
\(B>16K_{3}\) (see Remark~\ref{rem:constants}), and that
\(\epsilon(j-k+1/4)\le1\). 

Combining \eqref{eq:ugly}, \eqref{eq:est-firsty}, and
\eqref{eq:sec-term} finishes the proof of \eqref{eq:der-U(p,q)-bd} in
the case \(m=|\mu|\ge3\).

This finishes the proof of Lemma~\ref{lemanR-bis}.
\end{pf}
\appendix
\section{Multiindices and Stirling's Formula}\label{Notation}
 \renewcommand{\theequation}{A.\arabic{equation}}
 \renewcommand{\thetheorem}{A.\arabic{theorem}}
  \setcounter{equation}{0}  
  \setcounter{theorem}{0}  
We denote \(\mathbb{N}_{0}=\mathbb{N}\cup\{0\}\).
For \(\sigma=(\sigma_1,\sigma_2,\sigma_3)\in\N_0^3\)
we let \(|\sigma|:=\sigma_1+\sigma_2+\sigma_3\), 
and
\begin{align}
  \label{eq:diff-op}
  D^{\sigma}:=D_1^{\sigma_1}D_2^{\sigma_2}D_3^{\sigma_3}\ ,
  \quad
  D_{\nu}:={}-{\rm i}\frac{\partial}{\partial x_{\nu}}=:{}-{\rm
    i}\,\partial_{\nu}
  \ ,\quad \nu=1,2,3\,.
\end{align}
This way, 
\begin{align*}
  \partial^{\sigma}:=\frac{\partial^{|\sigma|}}{\partial {\bf
      x}^{\sigma}}:=\frac{\p^{|\sigma|}}{\p x_1^{\sigma_1}
   x_2^{\sigma_2}x_3^{\sigma_3}}= 
  (-{\rm i})^{|\sigma|}D^{\sigma}\,.
\end{align*}
We let
\(\sigma!:=\sigma_1!\sigma_2!\sigma_3!\),
and, for \(n\in\N_{0}\), 
\begin{align}
  \label{eq:mult-two}
  \binom{n}{\sigma}:=\frac{n!}{\sigma!}
  =\frac{n!}{\sigma_1!\sigma_2!\sigma_3!}\,. 
\end{align}
With this notation we have the multinomial formula, for 
\({\bx}=(x_1,x_2,x_3)\in\R^{3}\) and \(n\in\N_{0}\),
\begin{align}\label{eq:multi-nom}
   (x_1+x_2+x_3)^n=\sum_{\mu\in\N_{0}^{3},
     |\mu|=n}\binom{n}{\mu}{\bx}^{\mu}\,. 
\end{align}
Here,
\(\bx^{\mu}:=x_{1}^{\mu_1}x_2^{\mu_2}x_3^{\mu_3}\). 
It follows that
\begin{align}\label{eq:bound-sigma-fak}
   |\sigma|!\le 3^{|\sigma|}\sigma! \ \text{ for all }\sigma\in\N_{0}^{3}\,,
\end{align}
since, using \eqref{eq:mult-two}, that \((1,1,1)^{\mu}=1\) for all
\(\mu\in\N_{0}^{3}\),  
and  \eqref{eq:multi-nom}, 
\begin{align*}
   \frac{|\sigma|!}{\sigma!}=\binom{|\sigma|}{\sigma}
   \le\sum_{\mu\in\N_{0}^{3}, |\mu|=|\sigma|}
   \binom{|\sigma|}{\mu}(1,1,1)^{\mu}=(1+1+1)^{|\sigma|}=3^{|\sigma|}\,.
\end{align*}
We also define
\begin{align}
  \label{eq:mult}
  \binom{\sigma}{\mu}:=\frac{\sigma!}{\mu!(\sigma-\mu)!}
\end{align}
for \(\sigma,\mu\in\N_0^3\) with \(\mu\le\sigma\), that is,
\(\mu_{\nu}\le\sigma_{\nu}\), \(\nu=1,2,3\).
Note that for all \(\sigma\in\N_0^3\) and \(k\in\N_{0}\) (see
\cite[Proposition 2.1]{Kato-paper}), 
\begin{align}
  \label{eq:multiNom}
  \sum_{\mu\le\sigma, |\mu|=k}\binom{\sigma}{\mu}=\binom{|\sigma|}{k}\,.
\end{align}
Finally, by 
\cite[6.1.38]{AbraSte}, 
we have the following
generalization of Stirling's Formula: For \(m\in\N\), 
\begin{equation}\label{eq:Abra-Ste-binom}
  m! = \sqrt{2 \pi} m^{m+\frac12} \exp(-m+\frac{\vartheta}{12m})\
  \mbox{ for some }\ \vartheta=\vartheta(m) \in (0,1)\,,
\end{equation}
and so for \(n,m\in\N\), \(m<n\), 
\begin{align}\label{eq:est-binom-intro}\nonumber
  \binom{n}{m}&=  \frac{1}{\sqrt{2 \pi}}
  \frac{n^{n+1/2}}{m^{m+1/2}
  (n-m)^{n-m+1/2}} \exp(\frac{\vartheta(n)}{12
  n}-\frac{\vartheta(m)}{12 m}
  -\frac{\vartheta(n-m)}{12 (n-m)})\\ 
  & \leq  \frac{{\rm e}^{1/12}}{\sqrt{2\pi}}\,
  \frac{n^{n+1/2}}{m^{m+1/2}
  (n-m)^{n-m+1/2}}\,. 
\end{align}
\section{Choice of the localization}\label{localization}
 \renewcommand{\theequation}{B.\arabic{equation}}
 \renewcommand{\thetheorem}{B.\arabic{theorem}}
 \setcounter{equation}{0}  
 \setcounter{theorem}{0}  

Recall that, for \(\bx_{0}\in\R^{3}\setminus\{0\}\) and
\(R=\min\{1,|\bx_{0}|/4\}\), we have defined 
 \(\omega=B_{R}(\bx_{0})\),
\(\omega_{\delta}=B_{R-\delta}(\bx_{0})\), and that \(\epsilon>0\) is
such that \(\epsilon(j+1)\le R/2\).
Also, recall (see \eqref{def:Phi}) that  we have chosen a function
$\Phi$ (depending on \(j\)) satisfying  
\begin{equation}\label{def:Phi-BIS}
  \Phi \in C^{\infty}_{0}(\omega_{\epsilon (j+3/4)})\,,\quad
  0\le\Phi\le1\,,\quad
  \mbox{
  with }\; \Phi \equiv 1 \; \mbox{ on }\; \omega_{\epsilon(j+1)}\,. 
\end{equation}

For $j\in\N$ we choose functions $\{\chi_{k}\}_{k=0}^{j}$, and
$\{\eta_{k}\}_{k=0}^{j}$ (all depending on \(j\)) with the following
properties (for an illustration, see figures 1 and 2).
The functions $\{\chi_{k}\}_{k=0}^{j}$ are such that
\begin{align*}
  \chi_{0} &\in C^{\infty}_{0}(\omega_{\epsilon (j+1/4)})  
  \ \;\mbox{ with }\;  \ \,
  \chi_{0} \equiv 1 
  \; \ \ \,\mbox{ on }\ 
  \omega_{\epsilon(j+1/2)}\,,
  \intertext{and, for $k= 1, \dots, j$,}
  \chi_{k} \in &C^{\infty}_{0}(\omega_{\epsilon (j-k+1/4)})  \;
  \\&\mbox{ with }\;   
  \left\{\begin{array}{ll} 
  \chi_{k} \equiv 1 \; & \mbox{ on }\ \;
  \omega_{\epsilon(j-k+1/2)}\setminus
  \omega_{\epsilon(j-k+1+1/4)}\,,\\ 
  \chi_{k} \equiv 0 \; & \mbox{ on }\  \;\R^3
  \setminus ( \omega_{\epsilon(j-k+1/4)}\setminus
  \omega_{\epsilon(j-k+1+1/2)})\,, 
\end{array}\right.
\end{align*}
Finally, the functions $\{\eta_{k}\}_{k=0}^{j}$
are such that for $k= 0, \dots, j$, 
\begin{eqnarray*}
  \eta_{k} \in C^{\infty}(\R^3) & \; \mbox{ with }\; & 
  \left\{\begin{array}{ll} 
  \eta_{k} \equiv 1 
  \; & \mbox{ on }\ \;
  \R^3 \setminus
  \omega_{\epsilon(j-k+1/4)}\,,\\ 
  \eta_{k} \equiv 0 
  \; & \mbox{ on }\ \;
  \omega_{\epsilon(j-k+1/2)}\,. 
  \end{array}\right.
\end{eqnarray*}
Moreover we ask that 
\begin{equation}\label{chieta}
\begin{array}{lll}
  \chi_{0}+\eta_{0} \equiv 1 
  & \mbox{ on }\ \; \R^3,\\ 
  \chi_{k}+\eta_{k} \equiv 1 
  & \mbox{ on }\ \; \R^3 \setminus
  \omega_{\epsilon(j-k+1+1/4)}\ \mbox{ for }\ k=1,\dots,j\,,\\ 
  \eta_{k} \equiv \chi_{k+1}+\eta_{k+1} 
  & \mbox{ on }\ \;  \R^3 \ \mbox{
    for }\ k=0,\ldots, j-1\,. 
\end{array}
\end{equation}
Furthermore, we choose  these localization functions such that, for
a constant $C_{*}>0$ (independent of \(\epsilon, k, j, \beta\)) and
for all 
$\beta \in 
\N_{0}^3$ with $|\beta|=1$, we have that
\begin{equation}\label{eq:est-der-loc}
   |D^{\beta}\chi_{k}(\bx)| \leq \frac{C_{*}}{\epsilon}\, \
  \mbox{ and } \ \,  |D^{\beta}\eta_{k}(\bx)| \leq
  \frac{C_{*}}{\epsilon}\,, 
\end{equation}
for $k=0, \dots,j$, and all \(\bx\in\R^3\).


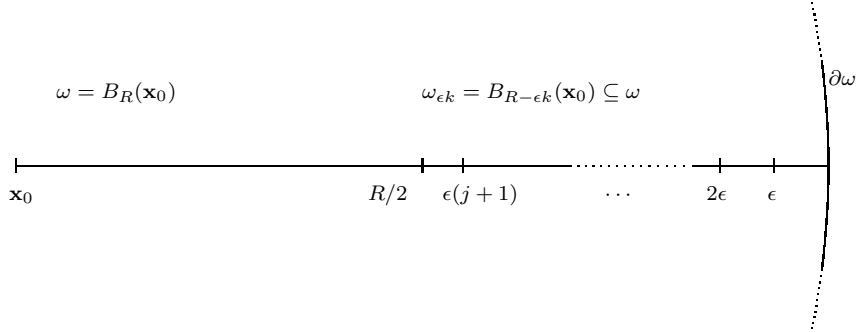
\begin{figure}[th]
\begin{center}
\setlength{\unitlength}{1.8cm}
{\tiny
\begin{picture}(6.4,2.4)(-.2,-1.2)
   \thinlines

\put(0,0){\line(1,0){4.1}}
\qbezier[15](4.1,0)(4.5,0)(5,0)
\put(5,0){\line(1,0){1}}

\put(5.2,-.05){\line(0,1){.1}}
\put(5.1,-.25){\mbox{$2 \epsilon$}}

\put(5.6,-.05){\line(0,1){.1}}
\put(5.55,-.25){\mbox{$\epsilon$}}

\put(0,-.05){\line(0,1){.1}}
\put(-.05,-.25){\mbox{$\bx_{0}$}}

\put(3,-.05){\line(0,1){.1}}
\put(2.6,-.25){\mbox{$R/2$}}

\put(3.3,-.05){\line(0,1){.1}}
\put(3.15,-.25){\mbox{$\epsilon(j+1)$}}

\put(4.35,-.25){\mbox{$\cdots$}}

\put(6,.6){\mbox{$\partial \omega$}}
\put(.3,.5){\mbox{$\omega=B_R(\bx_{0})$}}
\put(3,.5){\mbox{$\omega_{\epsilon k}=B_{R-\epsilon k}(\bx_{0})\subseteq\omega$}}

\qbezier(6,0)(6,0.38)(5.95,.77)
\qbezier(6,0)(6,-0.38)(5.95,-.77)
\qbezier[10](5.95,.77)(5.92,.999)(5.879,1.2)
\qbezier[10](5.95,-.77)(5.92,-.999)(5.879,-1.2)

\end{picture}
}
\caption{The geometry of \(\omega=B_R(\bx_{0})\) and the
  \(\omega_{\epsilon k}=B_{R-\epsilon k}(\bx_{0})\)'s.}

\end{center}
\end{figure}

\begin{figure}[ht]
\centering
\ersetze{a}{$\Phi$}
\ersetze{b}{$\chi_0$}
\ersetze{c}{$\chi_1$}
\ersetze{d}{$\chi_{j}$}
\ersetze{e}{$\epsilon(j+1)$}
\ersetze{f}{$\epsilon j$}
\ersetze{g}{$\epsilon(j-1)$}
\ersetze{h}{\ \ $\epsilon$}
\ersetze{i}{$\partial\omega$}
\ersetze{l}{$\eta_0$}
\ersetze{m}{$\eta_1$}
\ersetze{n}{$\eta_{j-1}$}
\ersetze{o}{$\eta_j$}
  \resizebox{12cm}{!}{\includegraphics{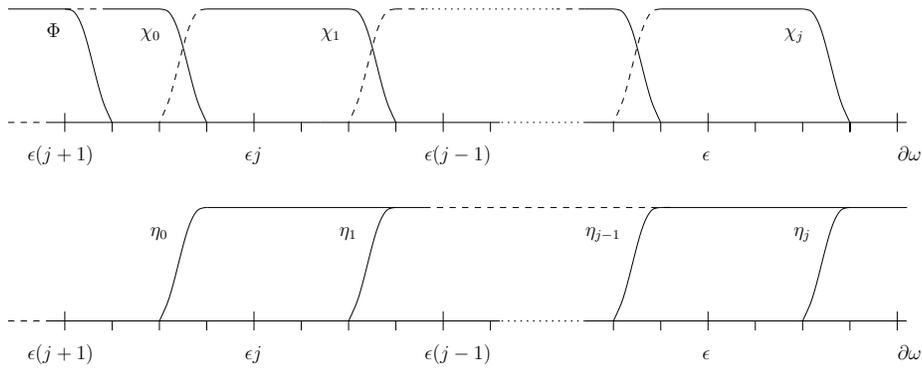}}
\caption{The localization functions.}
  \label{loc}
\end{figure} 

The next lemma shows how to use these localization functions.
\begin{lemma}\label{Edgardo}
For $j\in\N$ fixed, choose  functions $\{\chi_{k}\}_{k=0}^{j}$, and
$\{\eta_{k}\}_{k=0}^{j}$ as above, and
let $\sigma \in \N_{0}^3$ with
$|\sigma|=j$. For $\ell \in \N$ with $\ell \leq j$, choose multiindices
$\{\beta_{k}\}_{k=0}^{\ell}$ such that: 
\begin{equation*}
  |\beta_{k}|=k \mbox{ for }k=0, \dots, \ell\; ,\;  \beta_{k-1}<\beta_{k}
  \; \mbox{ for } \; k=1, \dots, \ell, \; \mbox{ and } \; \beta_{\ell}\leq
  \sigma\,. 
\end{equation*}

Then for all $g \in \mathcal{S}'(\R^3)$,
\begin{align}\label{eq:form-localization}
  D^{\sigma}g & =  \sum_{k=0}^{\ell} D^{\beta_{k}} \chi_{k}
  D^{\sigma-\beta_{k}} g \\ 
  & + \sum_{k=0}^{\ell-1} D^{\beta_{k}}[\eta_{k},D^{\mu_{k}}]
  D^{\sigma-\beta_{k+1}} g + D^{\beta_{\ell}}\eta_{\ell}
  D^{\sigma-\beta_{\ell}} g\,, 
  \nonumber
\end{align}
with $\mu_{k}=\beta_{k+1}-\beta_{k}$ for $k=0, \dots, \ell-1$ (hence,
\(|\mu_k|=1\)). 
\end{lemma}
\begin{proof}
We prove the lemma by induction on $\ell$ from $\ell=1$ to $\ell=j$.
We start by proving the claim for $\ell=1$. By using property
\eqref{chieta} of the localization functions and that
$\beta_1=\beta_{0}+\mu_{0}=\mu_{0}$ (since $\beta_{0}=0$) we find that
\begin{align}\label{e1}
  D^{\sigma}g
  =\chi_{0} D^{\sigma}g + \eta_{0}D^{\sigma}g
  =\chi_{0} D^{\sigma}g + \eta_{0}
  D^{\sigma-\beta_{1}+\mu_{0}}g\,.
\end{align}
The first term on the right side of \eqref{e1} is the term corresponding to
$k=0$ in the first sum in \eqref{eq:form-localization}. In the second term in
\eqref{e1}, commuting the derivative through \(\eta_0\),  we find that
\begin{align*}
  \eta_{0} D^{\sigma-\beta_{1}+\mu_{0}}g
  = D^{\mu_{0}}\eta_{0} D^{\sigma-\beta_{1}}g
  + [\eta_{0}, D^{\mu_{0}}]
  D^{\sigma-\beta_{1}}g\,.
\end{align*}
Since $\eta_{0}=\chi_{1}+\eta_{1}$ by property \eqref{chieta}, this
implies that
\begin{align}\label{e2}\nonumber
   \eta_{0} &D^{\sigma-\beta_{1}+\mu_{0}}g
   \\&=D^{\beta_{1}}\chi_{1} D^{\sigma-\beta_{1}}g
   +D^{\beta_{1}}\eta_{1} D^{\sigma-\beta_{1}}g
   + [\eta_{0},
   D^{\mu_{0}}] D^{\sigma-\beta_{1}}g\,.
\end{align}
The identity \eqref{eq:form-localization} for $\ell=1$ follows from
\eqref{e1} and \eqref{e2}. 

We now assume that \eqref{eq:form-localization} holds for $\ell-1$ for
some \(\ell\ge2\), i.e., 
\begin{align}\label{e3}
  D^{\sigma}g& = \sum_{k=0}^{\ell-1} D^{\beta_{k}}\chi_{k}
  D^{\sigma-\beta_{k}} g\\
   &{}\  + \sum_{k=0}^{\ell-2} D^{\beta_{k}}[\eta_{k},D^{\mu_{k}}]
  D^{\sigma-\beta_{k+1}} g +  D^{\beta_{\ell-1}}\eta_{\ell-1}
  D^{\sigma-\beta_{\ell-1}} g\,,\notag 
\end{align}
and prove it then holds for $\ell$. Since
$\beta_{\ell-1}=\beta_{\ell}-\mu_{\ell-1}$ we can rewrite the last term on the
right side of \eqref{e3} as  
\begin{eqnarray*}
  D^{\beta_{\ell-1}}\eta_{\ell-1} D^{\sigma-\beta_{\ell-1}} g =
  D^{\beta_{\ell-1}}\eta_{\ell-1} D^{\sigma-\beta_{\ell}+\mu_{\ell-1}}
  g\,.
\end{eqnarray*}
Again, commuting the $\mu_{\ell-1}$-derivative through
\(\eta_{\ell-1}\) this implies that
\begin{align}\notag
  D^{\beta_{\ell-1}}&\eta_{\ell-1} D^{\sigma-\beta_{\ell-1}} g
  \\\nonumber
  & =  D^{\beta_{\ell-1}+\mu_{\ell-1}}\eta_{\ell-1}
  D^{\sigma-\beta_{\ell}} g + D^{\beta_{\ell-1}}[\eta_{\ell-1},
  D^{\mu_{\ell-1}}] D^{\sigma-\beta_{\ell}} g \\ 
  & =  D^{\beta_{\ell}}(\eta_{\ell}+\chi_{\ell}) D^{\sigma-\beta_{\ell}} g +
  D^{\beta_{\ell-1}}[\eta_{\ell-1}, D^{\mu_{\ell-1}}]
  D^{\sigma-\beta_{\ell}} g\,,
  \label{e4} 
\end{align}
using \eqref{chieta}. 
Collecting together \eqref{e3} and \eqref{e4} proves that
\eqref{eq:form-localization} holds for \(\ell\). 

The claim of the lemma then follows by induction.
\end{proof}
\section{Norms of some operators on \(L^p(\R^3)\)}\label{app:smoothing}
\renewcommand{\theequation}{C.\arabic{equation}}
\renewcommand{\thetheorem}{C.\arabic{theorem}}
\setcounter{equation}{0}  
\setcounter{theorem}{0}  
In this section we prove two lemmas on bounds on certain operators
involving the 
operator \(E(\bp)=\sqrt{-\Delta+\alpha^{-2}}\).
\begin{lemma}\label{bounded-mult-op}
Let the operators \(S_{\nu}=E(\bp)^{-1}D_{\nu}\), \(\nu\in\{1,2,3\}\),
be defined for \(f\in\mathcal{S}(\R^3)\) by
\begin{align*}
  (S_{\nu}f)(\bx)=(2\pi)^{-3/2}\int_{\R^3}{\rm e}^{{\rm i}\bx\cdot\bp}
  E(\bp)^{-1}p_{\nu}\hat{f}(\bp)\,d\bp\,,
\end{align*}
with \(\hat{f}(\bp)=(2\pi)^{-3/2}\int_{\R^3}{\rm e}^{-{\rm
    i}\bx\cdot\bp} f(\bx)\,d\bx\) the Fourier transform of \(f\).
(Here, \(\bp=(p_1,p_2,p_3)\).)

Then, for all \(\mathfrak{p}\in(1,\infty)\), \(S_{\nu}\) extend to
bounded 
operators, \(S_{\nu}:L^{\mathfrak{p}}(\R^3)\to
L^{\mathfrak{p}}(\R^3)\), \(\nu\in\{1,2,3\}\). 
Clearly,
\(\|S_\nu\|_{\mathcal{B}_{\mathfrak{p}}}=\|S_\mu\|_{\mathcal{B}_{\mathfrak{p}}}\), 
\(\nu\ne\mu\). 
We let 
\begin{align}\label{norm-mult-op}
   K_{1}\equiv
   K_{1}(\mathfrak{p}):=
   \|S_1\|_{\mathcal{B}_{\mathfrak{p}}}\,.  
\end{align}
\begin{proof}
  This follows from \cite[Theorem 0.2.6]{Sogge} and the {\it Remarks}
  right after it. In fact, since (by induction),
\begin{align*}
  D_{\bp}^{\gamma}\big(p_{\nu}E(\bp)^{-1}\big)
  =P_{\gamma,\nu}(\bp) E(\bp)^{-1-2|\gamma|}\,,\quad \gamma\in\N_0^3\,,
\end{align*}
for some polynomials \(P_{\gamma,\nu}\) of degree \(|\gamma|+1\),
the functions \(m_{\nu}(\bp)=p_{\nu}E(\bp)^{-1}\) are smooth and satisfy the
estimates
\begin{align*}
  |D_{\bp}^{\gamma}m_{\nu}(\bp)|\le C_{\gamma,\nu}|\bp|^{-|\gamma|}\,,
  \quad \gamma\in\N_0^3\,,
\end{align*}
for some constants \(C_{\gamma,\nu}>0\), which is what is needed in
the reference above. 
\end{proof}

\end{lemma}
For \(\mathfrak{p},\mathfrak{q}\in[1,\infty]\), denote by
\(\|\cdot\|_{\mathcal{B}_{\mathfrak{p},\mathfrak{q}}}\) the 
operator norm on bounded operators from \(L^{\mathfrak{p}}(\R^{3})\)
to \(L^{\mathfrak{q}}(\R^{3})\). 

\begin{lemma}\label{normsmooth-Lp} 
For all \(\mathfrak{p},\mathfrak{r}\in[1,\infty)\), 
\(\mathfrak{q}\in(1,\infty)\), 
 with
\(\mathfrak{p}^{-1}+\mathfrak{q}^{-1}+\mathfrak{r}^{-1}=2\),
all
\(\alpha>0\), all \(\beta\in\N_0^3\) 
(with  \(|\beta|>1\) 
if \(\mathfrak{r}=1\)), and all
$\Phi, \chi \in C^{\infty}(\R^3)\cap
L^{\infty}(\R^3)$ with  
\begin{align}\label{eq:dist-supports} 
  \dist(\supp(\chi), 
  \supp(\Phi)) \geq d\,,
\end{align}
the operator \(\Phi E(\bp)^{-1}D^{\beta}\chi\) is bounded from
\(L^{\mathfrak{p}}(\R^3)\) to 
\((L^{\mathfrak{q}}(\R^3))'=
L^{\mathfrak{q}^{*}}(\R^3)\) (with
\(\mathfrak{q}^{-1}+{\mathfrak{q}^{*}}^{-1}=1\)), and  
\begin{align}\label{eq:smoothing-est-Lp}
  \|\Phi E(\bp)^{-1}&D^{\beta}\chi\|_{\mathcal{B}_{\mathfrak{p},\mathfrak{q}^{*}}}
  \\&
 \le 
  \frac{4\sqrt{2}}{\pi}
   \beta!\Big(\frac{8}{d}\Big)^{|\beta|}d^{3/\mathfrak{r}-2}
  \big(\mathfrak{r}(|\beta|+2)-3\big)^{-1/\mathfrak{r}}
  \|\Phi\|_{\infty}
  \|\chi\|_{\infty}\,.\nonumber
\end{align}

In particular, (when \(\mathfrak{r}=1\), i.e.,
\(\mathfrak{q}^{*}=\mathfrak{p}\)), 
\begin{align}\label{eq:smoothing-est-Lp-bis}
  \|\Phi E(\bp)^{-1}D^{\beta}\chi\|_{\mathcal{B}_{\mathfrak{p}}}
  \le 
  \frac{32\sqrt{2}}{\pi}
  \frac{\beta!}{|\beta|-1}\Big(\frac{8}{d}\Big)^{|\beta|-1}
  \|\Phi\|_{\infty}
  \|\chi\|_{\infty}\,,
\end{align}
for all \(\beta\in\N_0^3\) with  \(|\beta|>1\).
\end{lemma}
\begin{proof}
We use duality. Let \(f,g\in\mathcal{S}(\R^{3})\).
Note that, since \(\Phi f,D^{\beta}(\chi g)\in L^2(\R^3)\), the
spectral theorem, and the 
formula 
\begin{equation}\label{eq:formula x minus 1/2}
  \frac{1}{\sqrt{x}} = \frac{1}{\pi} \int_{0}^{\infty}
  \frac{1}{x+t}\,\frac{dt}{\sqrt{t}}\,,\ x>0\,,  
\end{equation}
imply that
\begin{equation*}
  (f,\Phi E(\bp)^{-1}D^{\beta} \chi  g ) = \frac{1}{\pi} \int_{0}^{\infty}
  \frac{dt}{\sqrt{t}} \,(f,\Phi
  (-\Delta+\alpha^{-2}+t)^{-1}D^{\beta} \chi g  )\,. 
\end{equation*}
By using the formula for the kernel of the operator
$(-\Delta+\alpha^{-2}+t)^{-1}$ \cite[(IX.30)]{RS2}, and
integrating by parts, we get that 
\begin{align*}
  &(f,\Phi E(\bp)^{-1}D^{\beta} \chi  g ) \\
  &= \frac{1}{\pi} \int_{0}^{\infty} \frac{dt}{\sqrt{t}} \int_{\R^3}
  \overline{f(\bx)}\Phi(\bx) \int_{\R^3}
  \frac{{\rm e}^{-\sqrt{\alpha^{-2}+t}\,|\bx-\by|}}{4\pi|\bx-\by|}
  \,[D^{\beta}(\chi  g)](\by) \,d\bx d\by\\ 
  &=  \frac{(-1)^{|\beta|}}{\pi}  \int_{0}^{\infty} \frac{dt}{\sqrt{t}}
  \int_{\R^3} \overline{f(\bx)}\Phi(\bx) 
  \int_{\R^3} 
  \Big(D^{\beta}_{\by}\, 
  \frac{{\rm e}^{-\sqrt{\alpha^{-2}+t}\,|\bx-\by|}}{4\pi|\bx-\by|} \Big)
  \chi(\by)  g(\by) \,d\bx d\by\,. 
\end{align*}
Notice that the integrand is different from zero only for
$|\bx-\by|\ge d$, due to the assumption \eqref{eq:dist-supports}.
Hence, by Fubini's theorem,
\begin{align}\label{eq:this}
  (f,\Phi E(\bp)^{-1}D^{\beta} \chi  g )=
  \int_{\R^3}\int_{\R^3}F(\bx)H(\bx-\by)G(\by)\,d\bx d\by\,,
\end{align}
with \(F(\bx)=\overline{f(\bx)}\Phi(\bx)\), \(G(\by)=\chi(\by)g(\by)\),
and
\begin{align}\label{eq:formula-H}\nonumber
  H(\bz)&\equiv H_{\alpha,\beta,d}(\bz)\\
  &={\1}_{\{|\,\cdot\,|\ge d\}}(\bz)\frac{(-1)^{|\beta|}}{\pi}\int_{0}^{\infty}
   \Big(D^{\beta}_{\bz}\,
  \frac{{\rm e}^{-\sqrt{\alpha^{-2}+t}\,|\bz|}}{4\pi|\bz|}\Big)
  \frac{dt}{\sqrt{t}} \,.
\end{align}
Now, by \eqref{eq:est-der-coul-bis} in Lemma~\ref{estkernel-BIS}
below, uniformly for \(\alpha>0\),  
\begin{align*}
  |H(\bz)|&\le {\1}_{\{|\,\cdot\,|\ge d\}}(\bz)
   \frac{\sqrt{2}}{4\pi^2}\frac{\beta!}{|\bz|}
   \Big(\frac{8}{|\bz|}\Big)^{|\beta|}
   \int_{0}^{\infty}{\rm e}^{-\sqrt{t}|\bz|/2}\,\frac{dt}{\sqrt{t}}
   \\&={\1}_{\{|\,\cdot\,|\ge d\}}(\bz)
   \frac{\sqrt{2}}{\pi^2}\frac{\beta!}{|\bz|^2}\Big(\frac{8}{|\bz|}\Big)^{|\beta|}\,,
\end{align*}
and so, for all \(\alpha>0\), \(\mathfrak{r}\in[1,\infty)\),  and all
\(\beta\in\N_{0}^{3}\) (with \(|\beta|>1\) if \(\mathfrak{r}=1\)), 
\begin{align*}
  \|H\|_{\mathfrak{r}}&\le
  (4\pi)^{1/\mathfrak{r}}\frac{\sqrt{2}}{\pi^2}\beta!\,8^{|\beta|}\Big(\int_{d}^{\infty}
  \big(|\bz|^{-|\beta|-2}\big)^r|\bz|^2\,d|\bz|\Big)^{1/\mathfrak{r}}
  \\&=(4\pi)^{1/\mathfrak{r}}\frac{\sqrt2}{\pi^2}
  \beta!\Big(\frac{8}{d}\Big)^{|\beta|}d^{3/\mathfrak{r}-2}
  \big(\mathfrak{r}(|\beta|+2)-3\big)^{-1/\mathfrak{r}}\,.
\end{align*}
From this, \eqref{eq:this}, and Young's inequality
\cite[Theorem~4.2]{LiebLoss} (notice that \(C_{Y}\le 1\)), 
follows that, with
\(\mathfrak{p},\mathfrak{q},\mathfrak{r}\in[1,\infty)\), 
\(\mathfrak{p}^{-1}+\mathfrak{q}^{-1}+\mathfrak{r}^{-1}=2\), 
\begin{align*}
  |(f,\Phi &E(\bp)^{-1}D^{\beta} \chi  g )|\le
   \|F\|_{\mathfrak{q}}\|H\|_{\mathfrak{r}}\|G\|_{\mathfrak{p}}\\&\le
   (4\pi)^{1/\mathfrak{r}}\frac{\sqrt2}{\pi^2}
  \beta!\Big(\frac{8}{d}\Big)^{|\beta|}d^{3/\mathfrak{r}-2}
  \big(\mathfrak{r}(|\beta|+2)-3\big)^{-1/\mathfrak{r}}
  \|F\|_{\mathfrak{q}}\|G\|_{\mathfrak{p}}
  \\&\le
  \frac{4\sqrt2}{\pi}\beta!\Big(\frac{8}{d}\Big)^{|\beta|}d^{3/\mathfrak{r}-2}
  \big(\mathfrak{r}(|\beta|+2)-3\big)^{-1/\mathfrak{r}}
  \|\Phi\|_{\infty}
  \|\chi\|_{\infty}\|f\|_{\mathfrak{q}}\|g\|_{\mathfrak{p}}\,.
\end{align*}
Since \(\mathcal{S}(\R^{3})\) is dense in both
\(L^{\mathfrak{p}}(\R^3)\) and \(L^{\mathfrak{q}^{*}}(\R^3)\), 
this finishes the proof of the lemma.
\end{proof}
\begin{lemma}\label{estkernel-BIS} 
  For all \(s>0\), \(\bx\in
  \R^3\setminus\{0\}\), and \(\beta\in\N_{0}^{3}\), 
  \begin{align}\label{eq:est-der-coul}
  \Big|\partial_{\bx}^{\beta}\frac{1}{|\bx|}\Big|
  &\le
  \frac{\sqrt{2}\beta!}{|\bx|}\Big(\frac{8}{|\bx|}\Big)^{|\beta|}\,, 
  \\
  \label{eq:est-der-coul-bis}
  \Big|\partial_{\bx}^{\beta}\frac{{\rm e}^{-s|\bx|}}{|\bx|}\Big| 
  &\le \frac{\sqrt{2} \beta !}{|\bx|}
  \Big(\frac{8}{|\bx|}\Big)^{|\beta|}{\rm e}^{-s|\bx|/2}\,. 
\end{align}
\end{lemma}
\begin{proof}
We will use the Cauchy inequalities \cite[Theorem 2.2.7]{Hor}.
To avoid confusion with the Euclidean norm \(|\cdot|\) (in \(\R^3\) or in
\(\C^3\)), we denote by \(|\cdot|_\C\) the absolut value in \(\C\).

Let, for \({\bf w}=(w_1,w_2,w_3)\in\C^3\) and \(r>0\),
\begin{align}\label{eq:polycylinder}
  P_{r}^3({\bf w})=\{\bz\in \C^3\ |\ |z_{\nu}-w_{\nu}|_{\C}<r\,,\ \nu=1,2,3\}
\end{align}
be the {\it poly-disc} with {\it poly-radius} \({\bf r}=(r,r,r)\).
The Cauchy inequalities then state that if $u$ is analytic in
$P_{r}^3({\bf w})$ and if $\sup_{\bz\in P_{r}^3({\bf w})}
|u(\bz)|_{\C} \leq M$, then   
\begin{align}\label{eq:Cauchy}
  |\partial_{\bf z}^{\beta} u({\bf w}) |_{\C} \leq M
  \beta!\,r^{-|\beta|} \quad \text{ for all }\beta\in\N_{0}^{3}\,.
\end{align}
We take ${\bf w} = {\bf x} \in \R^3\setminus\{0\} \subseteq \C^3$ 
and choose $r=|{\bf x}|/8$. We prove below that then we have (with
\(\bz^2:=\sum_{\nu=1}^3z_{\nu}^2\in\C\)) 
\begin{align}\label{eq:ineq-on-polycyl-new}
  {\rm Re}(\bz^2) \ge \frac12|\bx|^{2} \ \text{ for } \bz\in
  P^3_{r}(\bx)\,.
\end{align}
It follows 
that $\sqrt{{\bf
    z}^2} := \exp(\frac{1}{2} \Log\,{\bf z^2})$ is well-defined and
analytic on $P^3_r({\bf x})$ with $\Log$ being the principal branch of
the logarithm. 

We will also argue below that
\begin{align}
  \label{eq:norm-square-root-new}
  {\rm Re}(\sqrt{\bz^2})
  \geq \frac12|{\bf x}| \ \text{ for } \bz\in P^3_{r}(\bx)\,.
\end{align}
Then (by \eqref{eq:ineq-on-polycyl-new}) for all
\(\bz\in P^3_{r}(\bx)\),
\begin{align}\label{eq:new-one}
 |\sqrt{ {\bf z}^2}|_{\C} = \sqrt{ |{\bf z}^2|_{\C}} \geq \sqrt{ |
  {\rm Re }\,{\bf z}^2|} \geq |{\bf x}|/\sqrt{2}\,, 
\end{align}
and (by
\eqref{eq:norm-square-root-new}), for all $s \geq 0$ and all
\(\bz\in P^3_{r}(\bx)\),
\begin{align}\label{eq:new-label}
  |\exp( - s \sqrt{{\bf z}^2}) |_{\C} = \exp( -s {\rm
  Re}(\sqrt{\bz^2})) \leq \exp( -s |{\bf x}|/2 )\,. 
\end{align}
Therefore, \eqref{eq:est-der-coul} and \eqref{eq:est-der-coul-bis}
follow from \eqref{eq:Cauchy}, \eqref{eq:new-one}, and \eqref{eq:new-label}.

It remains to prove \eqref{eq:ineq-on-polycyl-new} and
\eqref{eq:norm-square-root-new}.

For \(\bz\in P^3_{r}(\bx)\), write \(\bz=\bx+\ba+{\rm i}\bb\)
with \(\ba,\bb\in\R^3\) satisfying
\(|z_{\nu}-x_{\nu}|_{\C}^2=a_{\nu}^2+b_{\nu}^2\le(|\bx|/8)^2\). Then
\begin{align*}
  \bz^2=|\bx+\ba|^2-|\bb|^2+2{\rm i}(\bx+\ba)\cdot\bb\,,
\end{align*}
so, with \(\epsilon=1/8\),
 \begin{align*}
   {\rm Re}(\bz^2)&=|\bx|^2+|\ba|^2+2\, \bx\cdot\ba -|\bb|^2
   \\&\ge (1-\epsilon)|\bx|^2+(2-\epsilon^{-1})|\ba|^2-(|\ba|^2+|\bb|^2)
   \\&\ge\frac{35}{64}|\bx|^2>\frac12|\bx|^2\,.
 \end{align*}
This establishes \eqref{eq:ineq-on-polycyl-new} .

It follows from \eqref{eq:ineq-on-polycyl-new} that, with \(\Arg\) the
principal branch of the argument,
\begin{align}\label{eq:Arg-new}
  {}-\frac{\pi}{4}\le \frac12\Arg(\bz^2)\le \frac{\pi}{4} \ \text{ for }
  \bz\in P^3_{r}(\bx)\,.
\end{align}
Furthermore (still
for \(\bz\in P^3_{r}(\bx)\)), because of \eqref{eq:Arg-new}, 
\begin{align}
  \label{eq:norm-square-root}
  {\rm Re}(\sqrt{\bz^2})
  =|\bz^2|_{\C}^{1/2}\cos(\tfrac12\Arg(\bz^2))
  \ge|\bz^2|_{\C}^{1/2}/\sqrt{2}\,.
\end{align}
Combining with \eqref{eq:ineq-on-polycyl-new} we get
\eqref{eq:norm-square-root-new}.

This finishes the proof of the lemma.
\end{proof}
\section{Needed results}\label{app:needed}
\renewcommand{\theequation}{D.\arabic{equation}}
\renewcommand{\thetheorem}{D.\arabic{theorem}}
\setcounter{equation}{0}  
\setcounter{theorem}{0}  
In this section we gather some results from the literature which are
needed in our proofs.
\begin{theorem}\cite[Theorem 5.8]{Adams}\label{adams}
Let \(\Omega\) be a domain in \(\R^n\) satisfying the cone
condition. Let \(m\in\N, \mathfrak{p}\in(1,\infty)\). 
If \(m\mathfrak{p}>n\), let \(\mathfrak{p}\le
\mathfrak{q}\le \infty\); if \(m\mathfrak{p}=n\), let 
\(\mathfrak{p}\le \mathfrak{q}<\infty\); if \(m\mathfrak{p}<n\), let
\(\mathfrak{p}\le \mathfrak{q}\le
\mathfrak{p}^*=n\mathfrak{p}/(n-m\mathfrak{p})\). Then  
there exists a constant \(K\) depending on \(m, n, \mathfrak{p},
\mathfrak{q}\) and the 
dimensions of the cone \(C\) providing the cone condition for
\(\Omega\), such that for all \(u\in W^{m,\mathfrak{p}}(\Omega)\),
\begin{align}\label{eq:adams}
   \|u\|_{L^{\mathfrak{q}}(\Omega)} 
   \le K\|u\|_{W^{m,\mathfrak{p}}(\Omega)}^\theta 
   \|u\|_{L^{\mathfrak{p}}(\Omega)}^{1-\theta}\,,
\end{align}
where \(\theta=(n/m\mathfrak{p})-(n/m\mathfrak{q})\).
\end{theorem}
We write \(K=K(m,n,\mathfrak{p},\mathfrak{q},\Omega)\).
We always use Theorem~\ref{adams} with \(n=3\), \(m=1\), and
\(\mathfrak{p}=p, \mathfrak{q}=3p\) for some \(p>3\). Hence
\(m\mathfrak{p}>n\), \(\mathfrak{p}\le\mathfrak{q}\le\infty\), and
\(\theta=\theta(p)=2/p<1\).
Moreover, we always use it 
with \(\Omega\) being a ball, whose radius in all cases is bounded from
above by \(1\) and from
below by \(R/2\) for some \(R>0\) fixed.

Let
\(K_0\equiv K_0(p)\equiv
K(1,3,p,3p, B_{1}(0))\) with \(B_{1}(0)\subseteq\R^3\) the unit ball
(which does satisfy the cone condition). Note that then, by
scaling, \eqref{eq:adams} implies that for all \(r\le1\) and all
\(\bx_{0}\in\R^3\),  
\begin{align}\label{eq:adams-bis}
  \|u\|_{L^{3p}(B_{r}(\bx_{0}))}
  \le K_0r^{-\theta}\|u\|_{W^{1,p}(B_{r}(\bx_{0}))}^{\theta}
  \|u\|_{L^{p}(B_{r}(\bx_{0}))}^{1-\theta}\,,
\end{align}
with \(\theta=2/p\).
 
To summarize, we therefore have the following corollary.
\begin{corollary}\label{cor:adams}
Let $p >3$ and $R\in (0,1]$. Then 
there exists a constant $K_2$, depending only on $p$ and $R$, such
that for all $r \in [R/2,1]$, \(\bx_0\in\R^3\), and all $u \in
W^{1,p}(B_{r}(\bx_0))$, 
\begin{align}
  \|u\|_{L^{3p}(B_{r}(\bx_{0}))}
  \le K_2\|u\|_{W^{1,p}(B_{r}(\bx_{0}))}^{\theta}
  \|u\|_{L^{p}(B_{r}(\bx_{0}))}^{1-\theta}\,,
\end{align}
with \(\theta =2/p\).
\end{corollary}
Here,
\begin{align}\label{eq:const-adams}
  K_2\equiv K_2(p,R)=(2/R)^{2/p}K_0(p)\,,
\end{align}
where \(K_{0}(p)=K(1,3,p,3p,B_{1}(0))\) in 
Theorem~\ref{adams} above. 
\begin{theorem}\cite[Theorem~4.2]{ChenWu}\label{elliptic}
Let \(\Omega\) be a bounded domain in \(\R^n\) and 
let \(a^{ij}\in C(\overline{\Omega})\), \(b^{i}, c\in
L^{\infty}(\Omega) \, \ i, j\in\{1,\ldots,n\}\), with
\(\lambda, \Lambda>0\) such that
\begin{align}\label{eq:cond-ell-1}
   &\sum_{i,j=1}^{n}a^{ij}\xi_i\xi_j\ge \lambda|\xi|^2\,,
   \ \text{ for all }\ x\in\Omega, \,\xi\in\R^n\,, \\  
   \label{eq:cond-ell-2}
   &\sum_{i,j=1}^{n}\|a^{ij}\|_{L^{\infty}(\Omega)}
  +\sum_{i=1}^{n}\|b^{i}\|_{L^{\infty}(\Omega)}
  +\|c\|_{L^{\infty}(\Omega)}\le\Lambda\,.
\end{align}
Suppose \(u\in W^{2,\mathfrak{p}}_{\rm loc}(\Omega)\) satisfies 
\begin{align}\label{eq:gen-elliptic}
  Lu=\sum_{i,j=1}^{n}{}-a^{ij}D_{i}D_{j}u
     +\sum_{i=1}^{n}b^{i}D_{i}u+cu=f\,.
\end{align}

Then for any \(\Omega'\subset\subset\Omega\),
\begin{align}\label{eq:elliptic-est}
   \|u\|_{W^{2,\mathfrak{p}}(\Omega')}\le C\big\{
   \frac{1}{\lambda}\|f\|_{L^{\mathfrak{p}}(\Omega)}
   +\|u\|_{L^{\mathfrak{p}}(\Omega)}\big\}\,,
\end{align}
where \(C\) depends only on \(n, \mathfrak{p}, \Lambda/\lambda,
\dist\{\Omega',\partial\Omega\}\), and the modulus of continuity of the
\(a^{ij}\)'s. 
\end{theorem}
We use Theorem~\ref{elliptic} in the case
  where \(\Omega'\) and \(\Omega\) are concentric balls 
(and with \(n=3\), \(\mathfrak{p}=3p/2\), \(a^{ij}=\delta_{ij},
b^{i}=c=0\); hence 
\(\Lambda=\lambda=1\)).
 Reading the proof of the theorem above with this case in mind (see
 \cite[Lemma~4.1]{ChenWu} in particular), one can make
 the dependence on \(\dist\{\Omega',\partial\Omega\}\)
 explicit. 
More precisely, we have the following corollary.
\begin{corollary}\label{rem:elliptic}
For all \(p>1\) 
there exists a constant \(K_{3}=K_{3}(p)\ge1\) such that
for all \(u\in W^{2,3p/2}(B_{r+\delta}(\bx_{0}))\) 
(with \(\bx_{0}\in\R^3, r,\delta>0\))
\begin{align}\label{eq:elliptic-est-bis}\nonumber
  &\|u\|_{W^{2,3p/2}(B_{r}(\bx_{0}))}
  \\&\qquad
  \le K_{3}\big\{\|\Delta u\|_{L^{3p/2}(B_{r+\delta}(\bx_{0}))}
  +\delta^{-2}\|u\|_{L^{3p/2}(B_{r+\delta}(\bx_{0}))}\big\}\,.
\end{align}
\end{corollary} 
\begin{theorem}\cite[Theorem~5, Section~5.6.2 (Morrey's
  inequality)]{Evans}\label{lemmaSobolev} 
Let \(\Omega\) be a boun\-ded, open subset in \(\R^n\), \(n\ge2\), and
suppose \(\partial\Omega\) is \(C^{1}\). Assume
\(n<\mathfrak{p}<\infty\), and \(u\in
W^{1,\mathfrak{p}}(\Omega)\). Then \(u\) has a version \(u^{*}\in
C^{0,\gamma}(\overline{\Omega})\), for \(\gamma=1-n/\mathfrak{p}\),
with the estimate
\begin{align}\label{eq:Morrey}
  \|u^{*}\|_{C^{0,\gamma}(\overline{\Omega})}\le
  K_{4}\|u\|_{W^{1,\mathfrak{p}}(\Omega)}\,. 
\end{align}
The constant \(K_{4}\) depends only on \(\mathfrak{p}, n\), and \(\Omega\).
\end{theorem}
Here, \(u^{*}\) is a version of the given \(u\) if \(u=u^{*}\) a.e..
Above, 
\begin{align}\label{eq:Holder-norm}
   \|u\|_{C^{0,\gamma}(\overline{\Omega})}
   :=\sup_{\bx\in\Omega}|u(\bx)|
   +\sup_{\bx,\by\in\Omega,
   \,\bx\neq\by}\frac{|u(\bx)-u(\by)|}{|\bx-\by|^{\gamma}}\,.
\end{align}
Of course,
\(\sup_{\bx\in\Omega}|u(\bx)|\le\|u\|_{C^{0,\gamma}(\overline{\Omega})}\). 
\begin{remark}\label{rem:Morrey}
 Note that
\cite[p.~245]{Evans} uses a definition of the 
\(W^{m,\mathfrak{p}}\)-norm which is slightly different from ours (see
\eqref{def:Sob-norm}), but which is an 
equivalent norm by equivalence of norms in finite dimensional
vectorspaces. Therefore, \eqref{eq:Morrey} holds with our definition
of the norm (but the constant \(K_{4}\) is not the same as the one in
\cite[Theorem~5, Section~5.6.2]{Evans}). 
\end{remark}
\begin{acknowledgement}
The authors thank Heinz Siedentop for suggesting to study this
problem.
AD, T{\O}S, and ES gratefully acknowledge a "Research in Pairs" (RiP) - stay
at the Mathematisches Forschungsinstitut Oberwolfach, where
parts of this research was carried out.  
SF was partially supported by the Lundbeck Foundation and the European
Research Council under the European Community's Seventh Framework
Programme (FP7/2007-2013)/ERC grant agreement n$^{\rm o}$ 202859. 
ES was partially  supported by the DFG (SFB/TR12). 
\end{acknowledgement}

  \bibliographystyle{amsplain}
  \bibliography{relHFanalytic}

\end{document}